 \documentclass[reqno,11pt,a4paper]{amsart}

\usepackage{amsfonts,amsmath,amssymb}
\usepackage[latin1]{inputenc}
\usepackage{dsfont}
\usepackage{enumerate}
\usepackage{xcolor}
\usepackage[all]{xy}
\def\notshow#1\notshowend{} 
\usepackage{hyphenat}
\usepackage{graphicx}

\usepackage{tikz}
\usepackage{tikz-3dplot}
\usepackage{pgfplots}

\usepackage{multirow}
\usepackage{array}

\newcommand{\C}{\mathcal{C}}

\usepackage{amssymb}
\usepackage{amsfonts}
\usepackage{mathrsfs}
\usepackage{hyperref}
\usepackage{enumitem}

\usepackage[normalem]{ulem} 

\def\bb#1\eb{{#1}}

\newcommand{\R}{\mathds R}
\newcommand{\N}{\mathds N}
\newcommand{\CC}{\mathds C}

\newtheorem{thm}{Theorem}[section]
\newtheorem{prop}[thm]{Proposition}
\newtheorem{lemma}[thm]{Lemma}
\newtheorem{cor}[thm]{Corollary}
\theoremstyle{definition}
\newtheorem{defi}[thm]{Definition}

\newtheorem{example}[thm]{Example}

\newtheorem{rem}[thm]{Remark}

\newcommand{\ben}{\begin{enumerate}}
\newcommand{\een}{\end{enumerate}}
\newcommand{\bit}{\begin{itemize}}
\newcommand{\eit}{\end{itemize}}
\newcommand{\edoc}{\end{document}}

\newcommand{\TT}{\mathrm{T}}
\newcommand{\HH}{\mathrm{H}}
\newcommand{\VV}{\mathrm{V}}
\newcommand{\fun}{\mathcal{F}}
\newcommand{\vect}{\mathfrak{X}}
\newcommand{\form}{\varOmega}
\newcommand{\ten}{\mathcal{T}}
\newcommand{\dv}{\dot{\partial}}
\newcommand{\lie}{\mathfrak{l}^\HH}
\newcommand{\Lie}{\mathfrak{L}}
\newcommand{\dd}{\mathrm{d}}
\newcommand{\vol}{d\mathrm{Vol}}
\newcommand{\di}{\mathrm{div}}
\newcommand{\la}{\mathrm{Lan}}

\title[Stress-energy tensor in Finsler spacetimes]{
On the  significance of 
the \\ stress-energy tensor  in Finsler spacetimes}

\author[M. A. Javaloyes]{Miguel Angel Javaloyes}\address{Departamento de Matem\'aticas, \hfill\break\indent Universidad de Murcia, \hfill\break\indent Campus de Espinardo,\hfill\break\indent 30100 Espinardo, Murcia, Spain} \email{majava@um.es}
\author[M. S\'anchez]{Miguel S\'anchez} \address{Departamento de Geometr\'{\i}a y Topolog\'{\i}a, Facultad de Ciencias \& IMAG, \hfill\break\indent Universidad de Granada,\hfill\break\indent Campus Fuentenueva s/n, \hfill\break\indent 18071 Granada, Spain}\email{sanchezm@ugr.es}
\author[F.F. Villase\~nor]{Fidel F. Villase\~nor} \address{Departamento de Geometr\'{\i}a y Topolog\'{\i}a, Facultad de Ciencias \& IMAG, \hfill\break\indent Universidad de Granada,\hfill\break\indent Campus Fuentenueva s/n, \hfill\break\indent 18071 Granada, Spain}\email{anbernal@no-gravity.eu}

\begin{document}

\begin{abstract} 

We revisit the physical arguments which lead to the definition of the stress-energy  tensor $T$ in the Lorentz-Finsler setting $(M,L)$ starting at classical Relativity. Both the standard heuristic approach using fluids and the Lagrangian one are taken into account. In particular, we argue that the Finslerian breaking of Lorentz symmetry makes $T$  an anisotropic 2-tensor  (i. e., a tensor for each $L$-timelike direction),  in contrast with {  the energy-momentum vectors} defined on $M$. Such a tensor is compared {  with different ones}   
obtained by using a Lagrangian approach.
  The notion of divergence is revised from a geometric viewpoint and, then, {  the conservation laws of $T$} for each observer field are revisited.    
   We introduce a natural {\em anisotropic Lie bracket derivation}, 
  which leads to a divergence 
obtained from the volume element and the non-linear connection associated with $L$ {  alone}.  The computation of this divergence  selects the Chern anisotropic connection, thus giving  a geometric interpretation  to previous choices in the literature.
\vspace{5mm}

\end{abstract}
\maketitle

\tableofcontents

\section{Introduction}
This article has a double  aim in Lorentz-Finsler Geometry. The first one is to revisit the physical grounds {  of} the stress-energy tensor $T$ \S \ref{s3}. The possible  extensions of the relativistic $T$ are discussed from the viewpoint of both fluids mechanics and Lagrangian systems. The second one is to {  revise} geometrically the notion of divergence \S \ref{sec:divergence vectors}, yielding consequences about the conservation of $T$ \S \ref{s5}. With this aim, we introduce {  new notions of} Lie   bracket and derivative  associated with a nonlinear connection and applicable to anisotropic tensors fields, which appear naturally in Finsler Geometry.

Finslerian modifications of General Relativity aim to find a tensor $T$
 collecting  the possible anisotropies in the distribution of energy, momentum and stress, which will serve as a source for the (now Lorentz-Finsler) geometry of the spacetime \cite{HPV0, HPV1, KSS, LiChang, SVV}. 
Some of these proposals may be waiting for experimental evidence, postponing then  how the basic relativistic notions would be affected. However, such a discussion is relevant to understand the scope and implications of the introduced Finslerian elements. In a previous reference \cite{BerJavSan}, the fundamentals of observers in the Finslerian setting were extensively studied, including its compatibility with {  the} Ehlers-Pirani-Schild approach. 
Now we focus on the stress-energy tensor $T$.

The difficulty to study such a $T$ is apparent. Recall that, using the principle of equivalence, General Relativity is reduced infinitesimally  into the Special one, which provides a background for interpretations. However, in the Lorentz-Finsler case, the infinitesimal model  is  changed into a Lorentz {\em norm} (instead of scalar product), implying a breaking of Lorentz invariance. This is a substantial issue in its own right which has been studied in the context of 
 {\em Very Special Relativity} and others  
\cite{Bogoslovsky, CG, GGP07, FP16,  Kos11}. As an additional difficulty, the infinitesimal  model changes with the point.\footnote{Berwald spaces \cite{FHPV,FPP18} are an exception, as the parallel transport becomes an isometry between the Lorentz norms. Thus, in some sense, these spaces would admit a principle of equivalence with respect to a Lorentz normed space (non-necessarily to Lorentz-Minkowski spacetime).}

  Two noticeable pre-requisites are the following: (a) only the value of the Lorentz-Finsler metric on causal directions is relevant \cite{BerJavSan, JavSan14} (this is  briefly commented   
  in the setup \S \ref{s2.3}), and (b)  there is a big variety of {  possible} extensions of the relativistic kinematic objects to the Finsler case, at least from the geometric viewpont (see the appendix \S \ref{s_A}).
Taking into account these issues,   the  extension of the notion of stress-energy tensor   to the Finslerian setting is discussed in \S \ref{s3}.

We start at the {\em fluids approach}.  As a preliminary question,   energy-momentum is discussed,  \S \ref{s3.1}. We emphasize that, even though {  this} is well-defined as a tangent vector {  in} each tangent space $\TT_pM$, $p\in M$, different observers $u$, $u'$ at $p$ will use  coordinates related by non-trivial linear transformations. Indeed, the latter will depend on both  $L$ and  the chosen way to measure  relative velocities. {  Moreover},  when {  the} stress-energy   $T$ is considered \S \ref{s3.2},   the arguments in Classical Mechanics and Relativity which support  its status  as a tensor   hold only partially in the {  Lorentz-Finsler} setting. Indeed, $T$  acquires a {  nonlinear} nature which  is codified in an (observer-dependent) anisotropic tensor, rather than 
{  in} a tensor on $M$.  

The Lagrangian approach is discussed in \S \ref{s_lagr}.  This approach has been developed recently  by  Hohmann, Pfeifer and Voicu \cite{HPV,HPV21}, who introduced an  {\em energy-momentum scalar function}. Here, we discuss the analogies and differences of this function with the canonical relativistic stress-energy tensor $\delta S_{matter}/ \delta  g^{\mu\nu}$ and the 2-tensor $T$ obtained from the fluids approach above. Relevant issues are the existence of different ways to obtain { } a 2-tensor starting at a scalar function, the recovery of this function from a matter Lagrangian and  the possibility to consider  the Palatini Lagrangian as the background one    (rather than Einstein-Hilbert type Lagrangians  used by the cited authors; recall that Palatini's becomes especially meaningful in the Finslerian case \cite{JSV2}).   The important case  of kinetic gases is considered explicitly ({Ex.} \ref{ex_kinetic}).

Once the definition of $T$ has been discussed, we focus on its conservation \S \ref{s5}, revisiting first the divergence theorem \S \ref{sec:divergence vectors}.  This is crucial in the Finslerian setting because, as  discussed before,  the Lagrangian approach above does not guarantee  a conservation law as the relativistic $\di(G)=0$.

 \S \ref{sec:divergence vectors} analyzes the divergence from a purely mathematical viewpoint. Now, $L$ is regarded as pseudo-Finsler (the results will be useful not only in any indefinite signature but also in the classical positive definite case) and $T$ will not be assumed to be symmetric a priori. Classically, the divergence of a vector field $Z$ is defined with the derivation associated with the Lie bracket $\left[Z,X\right]=\Lie_ZX$, {  applied to the} volume element. In the Finslerian case, however, {  the Lie derivative and bracket do} not make sense {  for} arbitrary anisotropic vector fields. 
 This difficulty was circumvented by 
 Rund \cite{Rund}, who redefined  $\mathrm{div}(Z)$  in such a way that a type of divergence theorem held. However,  the Lie viewpoint is  restored here.

\S \ref{subsec:bracket cartan} Once a nonlinear connection $\HH A$  (seen as a  horizontal distribution on $A$) is prescribed, we can define a Lie bracket $\lie_Z X$  and, then,  a Lie derivative $\Lie_Z^\HH X$ (Defs. \ref{d_4.1} and \ref{d_4.5}; Th. \ref{th:lie bracket} (C)). Noticeably, the former $\lie_Z $ is {  expressible} in terms of the infinitesimal flow of $Z$ (Prop. \ref{Lieflowformula}).

\S \ref{s_4.2} The divergence of $Z$ is naturally defined by using this     Lie bracket (Def. \ref{d_4.9}).   For the computation of $\mathrm{div}(Z)$, {  however}, one can use an {\em  anisotropic connection} $\nabla$ (this can be seen as a Finsler connection dropping its vertical part, see \S \ref{sec:preliminaries}) and {  a priori} Chern's one is not especially priviledged 
(Prop. \ref{prop:divergence vector characterization}).

\S \ref{sec:divergence tensor}. We give a general Finslerian version of {  the} divergence theorem for any anisotropic vector field $Z$, emphasizing the role of the choice of  an (admissible) vector field $V\colon M\rightarrow A$, which in the Lorentzian case can be interpreted as an observer field; this is expressed in terms of integration of forms in the spirit of {  Cartan's} formula (Th. \ref{th:divergence theorem}, {  Rem.} \ref{rem:normals}). 
We also explain how the boundary term can be expressed in different ways by using a normal either {  with} respect to {the pseudo-Riemannian metric} $g_V$ or to the {  fundamental tensor}, which were the choices {  of} Rund \cite{Rund} and Minguzzi \cite{Min17} resp.

 \S \ref{s5} gives  some applications to conservation laws. 

\S \ref{d_5.1}. First, we discuss the definition of divergence for the case of $T$. {  Our definition for vector fields was not biased to the Chern anisotropic connection,} but {  this} will be used for  $\mathrm{div}(T)$ (Def. \ref{d_divT}). The reason is that $\mathrm{div}(T)$ should  behave under contraction in a similar way as {  in} the isotropic case (namely, as in formula \eqref{eq:div(T(X))}), which privileges Chern's connection (Prop. \ref{prop:div(T(X))}).

\S \ref{d_5.2}. As an interlude about the appeareance of Chern's $\nabla$, 
a comparison with the possible use of Berwald's and previous approaches in the literature is done.

\S \ref{d_5.3}.
  A  conservation law for the flow of $T_V(X_V)$ is obtained (Cor. \ref{cor:conservation laws}), stressing  three hypotheses on the vanishing for $V$ of elements related to the stress-energy $T$ ($\mathrm{div}(T)=0$), the anisotropic vector $X$ ($\lie_X g=0$, generalizing the isotropic case) and {  a derivative of $V$}. The latter hypothesis is genuinely Finslerian and it means that  some terms related to  the nonlinear covariant derivative $\mathrm{D}V$ must vanish globally ($V$ can always be chosen such that they vanish at some point). {  It is worth pointing out that our general formula for the integral of the divergence \eqref{eq:divergence theorem} recovers the classical interpretation of the divergence as an infinitesimal growth of the flow (now observer-dependent). So, $\mathrm{div}(T)=0$ is equivalent to the conservation of energy-momentum in the instantaneous restspace of each observer, see Rem. \ref{rem:divT=0}.}

We finish by applying this general result to two examples.

First to Lorentz norms, showing that the conservation laws of Special Relativity still hold even though, now, the conserved quantity may be different for different observers. 
 As a second example, we give  natural conditions so that the flow of  $T_V(X_V)$ (whenever it exists as a Lebesgue integral, eventually equal to $\pm \infty$)  {is equal} in any two Cauchy hypersurfaces of a globally hyperbolic Finsler spacetime.  Indeed, we refine a previous result by Minguzzi \cite{Min17}, who assumed that $L$ was defined {  on} the whole $\TT M$  and  $T_V(X_V)$ was compactly supported.  We show that a combination of Rund's and Minguzzi's ways to compute the boundary terms allows one to obtain appropriate decay rates (namely, the properly Finslerian hypothesis \eqref{eq:decay condition}) which ensure the conservation.

\section{Preliminaries and setup} \label{sec:preliminaries}
First, let us set up some notation. In all the present text, $M$ is a connected smooth ($C^\infty$) manifold of dimension $n\geq 2$. As in previous references \cite{JSV1,JSV2}, any coordinate chart $(U,(x^1,...,x^n))$ of $M$ naturally induces a chart $(\TT U,(x^1,...,x^n,y^1,...,y^n))$ of $\TT M$ defined by the fact that
\[
v=y^i(v)\left.\frac{\partial}{\partial x^i}\right|_{\pi(v)}
\]
for $v\in\TT U$, where $\pi\colon \TT M\rightarrow M$ is the canonical projection. We abbreviate
\[
\frac{\partial}{\partial x^i}=:\partial_i,\qquad\frac{\partial}{\partial y^i}=:\dot{\partial}_i;
\]
these are vector fields on $\TT U$. At any rate, we will express our results in coordinate-free and geometric terms.

\subsection{Anisotropic tensors} We shall employ the framework of anisotropic tensors, following \cite{Jav1,Jav2,JSV1}, as it is simpler than previous ones. An open subset $A\subseteq\TT M$ with $\pi(A)=M$ is fixed; the elements $v\in A$ are called {\em observers}. We will denote by $\mathcal{T}^r_s(M_A)$ the space of (smooth) $r$-contravariant $s$-covariant $A$-anisotropic tensor fields ($r,s\in\N\cup\left\{0\right\}$), and by $\mathcal{T}(M_A):=\bigoplus_{r,s}\mathcal{T}^r_s(M_A)$ the full anisotropic tensor algebra. $\mathcal{F}(A)=\mathcal{T}^0_0(M_A)$ will be the space of functions on $A$. This time we will also put $\mathfrak{X}(M_A):=\mathcal{T}^1_0(M_A)$ for the space of anisotropic vector fields and $\varOmega_s(M_A)$ for the space of anisotropic $s$-forms (alternating anisotropic tensors, so that $\varOmega_1(M_A):=\mathcal{T}^0_1(M_A)$). The space $\ten(M)$ of classical tensor fields will be seen as a subspace of $\ten(M_A)$, formed by the {\em isotropic} elements, namely those which depend only on the point $p\in M$ and not on the observer at it. In particular, $\mathfrak{X}(M)\subseteq\vect(M_A)$. There is a distinguished element of $\vect(M_A)$: the {\em canonical (or Liouville) anisotropic vector field}, 
\[
\CC=y^i\,\partial_i,\qquad\CC_v:=v.
\]
For an open set $U\subseteq M$, we will put $\mathfrak{X}^A(U)$ for the set of (local) \emph{observer fields}, that is, those $V\in\mathfrak{X}(U)$ such that $V_p\in A\cap\mathrm{T}_p M$ for all $p\in U$. Given one of these and $T\in\mathcal{T}^r_s(M_A)$, their composition, denoted by  $T_V\in\mathcal{T}^r_s(U)$, makes sense. Finally, for $X\in\vect(M_A)$, there is also a canonical derivation $\dot{\partial}_X\colon\mathcal{T}^r_s(M_A)\rightarrow \mathcal{T}^r_s(M_A)$: the {\em vertical derivative along $X$},
\[
\left(\dot{\partial}_X T\right)_v:=\lim_{t\rightarrow 0}\frac{T_{v+t X_v}-T_v}{t}, \qquad \left(\dot{\partial}_X T\right)_{j_1,...,j_s}^{i_1,...,i_r}=X^{j_{s+1}}\dot{\partial}_{j_{s+1}}T_{j_1,...,j_s}^{i_1,...,i_r}.
\]

\subsection{Nonlinear and anisotropic connections}
In this article, a {\em nonlinear connection on $A\rightarrow M$} is defined as a (horizontal) subbundle $\HH A\subseteq \TT A$ such that $\TT A=\HH A \oplus \VV A$, where $\VV A:=\left.\mathrm{Ker}(\dd \pi)\right|_A$ is the vertical subbundle. For other options and the rudiments, see \cite{JSV1}. Nonlinear connections are characterized by their {\em nonlinear coefficients} $N^i_j$, 
\begin{equation}
	\HH_vA = \mathrm{Span} \left\{\left.\delta_i \right|_v \right\}, \qquad \delta_i := \frac{\delta}{\delta x^i} := \frac{\partial}{\partial x^i}  - N_i^j \frac{\partial}{\partial y^j},
	\label{eq:nonlinear connection}
\end{equation}
and also by their {\em nonlinear covariant derivative} $\mathrm{D}_X\colon\mathfrak{X}^A(U)\rightarrow\mathfrak{X}(U)$, 
\begin{equation}
\mathrm{D}_XV:=X^j\left(\frac{\partial V^i}{\partial x^j}+N^i_j(V)\right)\partial_i,
\label{eq:nonlinear covariant derivative}
\end{equation}
for $X\in\mathfrak{X}(U)$. They also provide (at least locally) a {\em nonlinear parallel transport} of observers $v\in A\cap\TT_{\gamma(0)} M$ along curves $\gamma\colon\left[0,t\right]\rightarrow M$. Namely, a map $P_t:A_{\gamma(0)}\rightarrow A_{\gamma(t)}$ defined as $P_t(v)=V(t)$, being $V$ the only vector field along $\gamma$ such that $V(0)=v$ and $D_{\dot\gamma}V=0$ (see \cite[Def. 12]{JSV1} and the comment below).

An {\em $A$-anisotropic connection} is an operator $\nabla\colon\vect(M)\times\vect(M)\rightarrow\vect(M_A)$ satisfying the usual Koszul derivation properties, see \cite{Jav1,Jav2,JSV2}. In a chart domain $U$, they are characterized by their Christoffel symbols 
$\varGamma_{jk}^i:A\cap TU\rightarrow\R$,
\[
\nabla_{\partial_j}\partial_k =: \varGamma_{jk}^i \partial_i.
\]
{  They can be seen as vertically trivial} linear connections on the vector bundle $\VV A \rightarrow A$ \cite[Th. 3]{JSV1}. On the other hand, every anisotropic connection has an {\em underlying nonlinear connection}, the only one with nonlinear coefficients
\[
N^i_j:=\varGamma_{jk}^i y^k.
\]
As a consequence, they define the {\em covariant derivative} $\nabla\colon\ten_s^r(M_A)\rightarrow\ten_{s+1}^r(M_A)$ for any anisotropic tensor: 
\[
\nabla_{j_{s+1}}T^{i_1,...,i_r}_{j_1,...,j_s}=\delta_{j_{s+1}} T^{i_1,...,i_r}_{j_1,...,j_s}+\sum_{\mu=1}^r\varGamma_{j_{s+1}k}^{i_\mu}T^{i_1,...,k,...,i_r}_{j_1,...,j_s}-\sum_{\nu=1}^s\varGamma_{j_{s+1} j_\mu}^{k}T^{i_1,...,i_r}_{j_1,...,k,...,j_s}.
\]

\subsection{Lorentz-Finsler metrics}\label{s2.3}

From now on, we will always assume that $A$ is conic ($\lambda v\in A$ for $v\in A$ and $\lambda\in (0,\infty)$). We shall follow the definitions and conventions in \cite{JS20,JSV1}. In particular, a {\em Finsler spacetime} $(M,L)$ is a (connected) manifold $M$ endowed with a {\em (properly) Lorentz-Finsler metric} 
$L\colon\overline{A} \subseteq \TT M\setminus\mathbf{0} \rightarrow \left[0,\infty\right)$. $L$ is required to be  
smooth, positive homogeneous and, when restricted to each $A_p:= \TT_pM\cap A$ ($p\in M$), its vertical Hessian $g$ is non-degenerate with signature $(+,-,\dots , -)$;  $A_p$ must be connected and salient, and  its boundary  in $\TT M\setminus\mathbf{0}$, which must be equal to $L^{-1}(0)$, is a (strong) {\em cone structure} $\C$. In particular, at each point $p$, $L$ is a {\em Lorentz norm}. By positive homogeneity, $L$ is determined by its {\em indicatrix} $L^{-1}(1)$. 

Notice that the cone $\C$ yields a natural notion of timelike, lightlike and spacelike tangent vectors but $L$ is not defined on the latter. Indeed, we are not interested in the value of $L$ on spacelike vectors by  physical reasons  which  are analyzed in \cite{BerJavSan}. Roughly, only particles (massive, massless) can be measured and, so, experimental evidences  only can affect 

$\Sigma$ and  
$\C$. Even though this also happens in 
classical Relativity,  the value of the Lorentz metric on the (future-directed) timelike vectors is enough to extend it to all the directions. Indeed, the anisotropies in Finsler spacetimes should be regarded as originated by the distribution of matter and energy in the causal directions rather than by (unobservable) spacelike anisotropies.     

 Even though it is the Lorentz-Finsler case which has a physical interpretation, in all other aspects the theory carries on if $L$ is just {\em pseudo-Finsler}, namely positively $2$-homogeneous with non-degenerate $g$ on $A$. In fact, this is the context in which we will develop \S \ref{sec:divergence vectors} and \ref{s5}, as they are of a more mathematical character.

The {\em Cartan tensor} of $L$ is
\[
C:=\frac{1}{2}\dot{\partial g}, \qquad C_{ijk}=\frac{1}{2}\frac{\partial g_{ij}}{\partial y^k}.
\]
It is actually symmetric, so one can define the {\em mean Cartan tensor} as
\begin{equation}
C^\mathfrak{m}(X):=\mathrm{trace}_g\left\{C(X,-,-)\right\},\qquad\left(C^\mathfrak{m}\right)_j=g^{ik}C_{ijk}=:C_j,	
\label{eq:mean cartan}
\end{equation}
for $X\in\vect(M_A)$. $L$ has also a canonically associated connection: the {\em metric nonlinear connection}, $\HH A$, of nonlinear coefficients
\begin{equation}
N^i_j:=\gamma_{jk}^iy^k-C^i_{jk}\gamma^k_{ab}y^ay^b,\qquad\gamma_{jk}^i:=\frac{1}{2}g^{ic}\left(\frac{\partial g_{cj}}{\partial x^k}+\frac{\partial g_{ck}}{\partial x^j}-\frac{\partial g_{jk}}{\partial x^c}\right).
\label{eq:metric connection}
\end{equation}
This is the underlying nonlinear connection of several anisotropic connections. One is the (Levi-Civita)--Chern $\nabla$, the only symmetric anisotropic connection that parallelizes $g$. It is the horizontal part of Chern-Rund's and Cartan's classical connections and it has Christoffel symbols
\begin{equation}
\varGamma_{jk}^i:=\frac{1}{2}\,g^{il}\left(\frac{\delta g_{lj}}{\delta x^k}+\frac{\delta g_{lk}}{\delta x^j}-\frac{\delta g_{jk}}{\delta x^l}\right),
\label{eq:chern}
\end{equation}
 where the $\delta_i$ are those associated with \eqref{eq:metric connection}. Another one is the Berwald $\widehat{\nabla}$. This is the horizontal part of Berwald's and Hashiguchi's classical connections and it has Christoffel symbols
 \begin{equation}
 \widehat{\varGamma}_{jk}^i:=\frac{1}{2}\,g^{il}\left(\frac{\delta g_{lj}}{\delta x^k}+\frac{\delta g_{lk}}{\delta x^j}-\frac{\delta g_{jk}}{\delta x^l}\right)+\la_{jk}^i.
 \label{eq:berwald}
 \end{equation}
Here, $\la_{jk}^i$ are the components of a tensor metrically equivalent to the {\em Landsberg tensor of $L$}, which, among many other ways, can be defined as
\[
\la_{ijk}:=\frac{1}{2}g_{lm}\dot{\partial}_i\dot{\partial}_jN_k^ly^m
\]
for the $N_k^l$ of \eqref{eq:metric connection} (see \cite[(37)]{Jav1}). The Landsberg tensor is actually symmetric too, so one can define the {\em mean Landsberg tensor of $L$} as
\begin{equation}
\la^\mathfrak{m}(X):=\mathrm{trace}_g\left\{\la(X,-,-)\right\},\qquad\left(\la^\mathfrak{m}\right)_j=g^{ik}\la_{ijk}=:\la_j.	
\label{eq:mean landsberg}
\end{equation}

\section{Basic interpretations on the stress-energy tensor $T$ 
}\label{s3}

Let us start with a discussion at each event $p\in M$ of a Finsler spacetime $(M,L)$. We can consider $\TT_pM$ endowed with the Lorentz norm $L|_{\TT_pM}$.  In most of this section, the discussion relies essentially on the particular case  when  $M$ is a real affine $n$-space with associated vector space $V$ (which plays the role of $\TT_pM$ in the general case) and $L$ is a Lorentz-Finsler norm on $V$ with indicatrix $\Sigma$ and cone $\C$ included in $V$.  Given  $u, u'\in \Sigma$, consider the corresponding fundamental tensors   $g_u$ and $g_{u'}$ and take orthonormal bases $B_u$, $B_{u'}$, obtained extending $u$, $u'$.  In a natural way, these bases live in $\TT_uV, \TT_{u'}V$ and they can be identified with bases in $V$ itself. Assuming this, the change of coordinates between $B_u$, $B_{u'}$ is linear but  not a Lorentz transformation, in general.

Extending the interpretations in Relativity, $p\in M$ is an {\em event},  the affine simplification  includes  the case of  Very Special Relativity \cite{Bogoslovsky, CG, GGP07},  $u\in \Sigma$  can be regarded as an {\em  observer}, the tangent space to the indicatrix $\TT_u\Sigma$ (i.e., the subspace $g_u$-orthogonal to $u$ in $\TT_uV\equiv V$) becomes the {\em restspace} of the observer $u$, and $B_u$ is an {\em inertial reference frame} for this observer. The {\em Lorentz invariance breaking} corresponds to the fact that the bases $B_u$ and $B_{u'}$ are orthonormal for the different metrics $g_u, g_{u'}$ and, thus, the linear transformation between the coordinates of $B_u$ and $B_{u'}$  (when regarded as elements of the same vector space $\TT_uV\equiv V \equiv \TT_{u'}V$) is not a Lorentz one. 
If the affine simplification is dropped, such elements (observers, restspaces)  must be regarded as  instantaneous  at $p\in M$. 

It is worth emphasizing that, according to the viewpoint introduced in \cite{JavSan14} and discussed extensively in \cite{BerJavSan}, the spacelike directions are not physically relevant for the Lorentz-Finsler metric. However, each (instantaneous) observer does have a restspace with a Euclidean scalar product. In the case of classical Relativity, Lorentz-invariance permits  natural identifications between these restspaces, and they become consistent with the value of the scalar product on spacelike directions. Certainly, a Lorentz norm $L$ 

could be extended outside these directions (maintaining the Lorentz signature for its fundamental tensor) but this can be done  in many different ways, and no relation with the scalar products $g_u, u\in \Sigma$ would hold.

The dropping of natural identifications associated with the Lorentz invariance implies that many notions which are  unambiguously  defined in  classical Relativity
admit many different alternatives now. In the Appendix we analyze some of them for the {\em relative velocity} between observers as well as other kinematical concepts. This is  taken  into account in the following discussion about how the  Finslerian setting affects the notion of {\em energy-momentum-stress} tensor.

\subsection{ 
Particles and dusts: anisotropic picture of isotropic elements} \label{s3.1}
In principle, there is no reason to modify the classical relativistic interpretation of $p= m u$ as the {\em (energy-) momentum vector}  of a particle of {\em (rest) mass} $m>0$ moving in the observer's direction $u\in \Sigma$. Moreover, if the particle moves in such a way that $m$ is constant, it will be represented by a unit timelike curve $\gamma(\tau)$ such that $p(\tau )=m\gamma'(\tau)$ will be its instantaneous momentum at each {\em proper time} $\tau$. The  (covariant)  derivative $p'=m\gamma''$ would be the force $F$ acting on the particle, which is necessarily $g_{\gamma'}$-orthogonal to $\gamma'$ (i.e., the force lies in the instantaneous restspace  of the particle). Then, the relativistic conservation of the momentum in the absence of external forces would retain its natural meaning, namely, if the particle represented by $(m, \gamma)$ splits into two $(m_1, \gamma_1)$ and $(m_2, \gamma_2)$ at some $\tau_0$ then $m\gamma'(\tau_0)= m_1\gamma_1'(\tau_0)+m_2\gamma_2'(\tau_0)$. 

The Appendix suggests that the way how an observer $u$ may measure the 
energy-momentum and conservation may be non-trivial. In particular, if one assumes that an observer $u$ measures $m\gamma' \in \TT_pM$ by using a $g_u$-orthonormal basis $B_u$

in general, $g_u  (m\gamma', m\gamma')\neq m^2 (=L(m\gamma'))$. Moreover,  as we have already  commented,  the coordinates for other observer $u'$ will not transform by means of Lorentz transformation. However, as the transformation of their coordinates is still linear, and both of them will write consistently $m\gamma'(\tau_0)= m_1\gamma_1'(\tau_0)+m_2\gamma_2'(\tau_0)$  in their coordinates.

Particles are also the basis to model dusts, which constitute  the simplest class of relativistic fluids. A dust is represented by a  
number-flux vector field  $N=nU$,  where $U$ represents the intrinsic velocity of the particle in the dust, i.e. a comoving observer, and $n$ is the  density of the dust for each momentaneously comoving reference frame. Comparing with the case of energy momentum, $N$ is also an intrinsic object which lives at the tangent space of each point and $U$ gives the priviledged observer who measures $n$. However, the  measures of $n$ by different observers involve different measures of the volume. As explained in the Appendix, the length contraction may be fairly unrelated to the relative velocities of the observers. This implies a more complicated transformation of the coordinates by different observers.  Anyway,  the transformations between these coordinates would remain linear and, so,  they could still agree in the fact that they are measuring the same intrinsic vector field.

Summing up, in the case of both particles and dusts, one assumes that the physical property lives in $V$ (or, more properly, in each tangent space $\TT_pM$ of the affine space) and there is a priviledged (comoving) observer $u$. The transformation of coordinates for other observer $u'$ may be complicated   but, at the end, it is a linear transformation which can be determined by specifying 
 the geometric quantities which are being measured as well as  the geometry of $\Sigma$. Thus, by using the coordinates measured by each observer  one could construct and  anisotropic vector field at each $p\in M$, which will fulfill some constraints, as  the measurement by one of the observers (in particular, the priviledged one) would determine  the  measurements by all the others.

\subsection{Emergence of an anisotropic stress-energy tensor}\label{s3.2}
The situation, however, is subtler for more general fluids, which are modelled classically by a 2-tensor on the underlying manifold. 

Let us start recalling the Newtonian and Lorentzian cases. In Classical Mechanics one starts working in an orthonormal basis of  Euclidean space to obtain the components $T_{ij}$ of the Cauchy stress tensor, which   give the flux of $i$-momentum (or force) across the 
$j$-surface in the background\footnote{\bb In this section, $i,j=1,2,3$ and $\mu,\nu=0,1,2,3$, but in the others they will run freely from 1 to $n$ ($=$ dim $M$). \eb}. The  laws of conservation of linear momentum and static equilibrium of  forces 
imply that these components give truly a 2-tensor (linear in each variable) and the conservation of linear momentum implies that this tensor is symmmetric. 

In the relativistic setting, each observer will determine some symmetric components  $T^{ij}$ in its restspace by essentially the same procedure as above. Additionally, it constructs $T^{00}$,  $T^{0i}$ and $T^{i0}$ as the density energy, energy flux across 
$i$-surface and $i$-momentum density, resp. The interpretation of these magnitudes completes the symmetry\footnote{\label{foot_sym}The symmetry of $T$ is dropped for the case of theories with high spin because of its contribution to angular momentum. 

} $T^{0i}=T^{i0}$ as well as the linearity in the $0$-component. 
However, the bilinearity in the components $T^{\mu\nu}$ has been only ensured for vectors in the restspace of the observer. In Relativity, one can claim Lorentz invariance in order to complete the reasons justifying that, finally, the components $T^{\mu\nu}$ will transform as a tensor\footnote{See for example  \cite[\S 4.5]{Sc}, \cite[\S 35]{LL}.}. 

Nevertheless, it is not clear in Lorentz-Finsler geometry why the transformation of the components $T_{ij}$ from an observer $u$ to a second one $u'$ must be linear, taking into account that they apply to  spacelike coordinates in distinct Euclidean  subspaces  and no Lorentz-invariance is assumed. 
Indeed, the following simple academic example shows that this is not the case.

\begin{example}\label{ex_Lioville stress energy tensor} Assume that $(M,L)$ is an affine space with a Lorentz norm with domain $A$  and consider the anisotropic tensor\footnote{The division by $L$ is so that $\mathbf{T}$ is $0$-homogeneous overall, as anisotropic stress-energy tensors should be in order to correctly generalize the classical case.} $\mathbf{T}  
=L^{-1}\phi \; \CC 
\otimes \CC  
$, where $\CC$ is the canonical (Liouville) vector field and $\phi: \Sigma \rightarrow \R$ is a smooth function which is extended as a 0-homogeneous function on $A$. Then, for each $u\in\Sigma$ and $w\in \TT_u\Sigma$ one has $\mathbf{T}_u(u,u)=\phi(u)$, $\mathbf{T}_u(w,w)=0$, 
$\mathbf{T}_u(u,w)=0$. In this case, each 
$\mathbf{T}_u$ is a symmetric 2-tensor, but the information on $\mathbf{T}$ requires the knowledge of $\phi(u)$ for all possible $u\in\Sigma$. Recall that this example holds even if $(M,L)$ is  the Lorentz-Minkowski spacetime regarded as a Finsler spacetime (but no Lorentz-invariance is assumed for $\mathbf{T}$).
\end{example}

Therefore, the following issues about $T$ appear:    

\bit
 \item[(a)] \label{i_a} Observer dependence:  even if we assume  that the components $T^{\mu\nu}$ measured by any observer $u$ are bilinear and then,  it is a standard tensor, the components measured by a second observer $u'$ may transform by a linear map which depends on $\Sigma$ as well as the experimental way of measuring (as in the case of the energy-momentum vector).

 \item[(b)] Nonlinearity: it is not clear even why such a linear transformation must  exist,  as bilinearity is only ensured in the direction of $u$ and of its restspace. Thus, the tensor $T_u$  measured by a single observer $u$ would not be enough to grasp the physics of the fluid at each event $p\in M$, as in the example above.

 \item[(c)] \label{i_c} Contribution of the anisotropies of $\Sigma$: as an additional possibility,  the local geometry of $\Sigma$ at $u$   underlies  the measurements of this observer and might provide a contribution for  the stress-energy tensor itself. 
 \eit
 
Summing up, Lorentz-Finsler geometry leads to assume that the measurements by $u$ are not enough to determine the state of the fluid and  the stress-energy tensor should be regarded as a non-isotropic tensor field, determined by the measurements of all the observers.

Formally, this means  an  {\em anisotropic tensor} $T\in \mathcal{T}_0^2(M_A)$ (see  \cite{JSV1}  for a summary of the formal approach), which can be expressed locally as
$$
T_v= T^{\mu\nu}(v)
\left. \partial_{\mu}\right|_{x} 
\otimes
\left.\partial_{\nu}\right|_{x}, \qquad  v = y^\mu \left. \frac{\partial}{\partial x^\mu}\right|_x \equiv (x,y)\in A\subset \TT M,
$$
where $T^{\mu\nu}(\lambda v)=T^{\mu\nu}(v)$ for all $\lambda>0$ (i.e. $T_v$ depends only on the direction of $v$). As a first approach (recall footnote \ref{foot_sym}), we can assume $T^{\mu\nu}=T^{\nu\mu}$. Consistently, we will assume that there exists a Lorentz-Finsler metric $L$ on $M$ with indicatrix $\Sigma \subset \TT M$ and, so, indexes can be raised and lowered by using its fundamental tensor $g$. The fact that $T$ has order 2 is important to establish classical analogies. However, other  tensors might appear as more fundamental energy-momentum tensors and, then, one would try to derive a semi-classical 2-tensor as in \S \ref{s_lagr}.  

 In principle, the intuitive relativistic interpretations  would be transplanted directly to each $v$, whenever  $v\in\Sigma$.   That is, given two $g_v$-unit vectors $u,w$, the value $T_v(u,w)$ of the 2-covariant  stress-energy tensor    perceived by the observer $v$ (at $x=\pi(v)$)  is obtained as the flux of $w$-energy-momentum per unit of $g_v$-volume  orthogonal to $u$. More precisely, let      $B(u)$ be a  small coordinate 3-cube in a hypersurface  $g_v$-orthogonal to $u$ and  $P_B$ is the total flux of the energy-momentum of particles crossing $B(u)$ (being positive from the $-u$ side to the $u$ side and negative the opposite direction),  then the $w$-energy-momentum per unit of $g_v$-volume is
\[  \epsilon  \, T_v(u,w):=  \lim_{Vol_{g_v}(B(u))\rightarrow 0}\frac{g_v(P_B, w )}{Vol_{g_v}(B(u))}.\]

 where $\epsilon=g_v(w,w)$.  As a Finslerian subtlety, recall that  $g_v$ is only defined \bb in $T_v(T_xM)$ and then in \eb $T_xM$ (i.e., it is trivially extended to $B(u)$ in a coordinate depending way), but  the above limit depends only on the value of $g_v$. Namely, if one considers two semi-Riemannian metrics $g$ and $\tilde g$ in a neighborhood of $p$ such that $ g_p=\tilde g_p$ and $B_n$ are open subsets with $p$ in the interior of $B_m$ for all $n\in\N$ and  $\lim_{n\rightarrow +\infty}vol_g(B_m)=0$,  then 
\[\lim_{m\rightarrow +\infty}\frac{vol_g(B_m)}{vol_{\tilde g}(B_m)}=1.\]
In particular, we have the interpretations \bb (recall signature $(+, -, -, -)$): \eb
\begin{enumerate}
\item $T_v(v,v)$ is the energy density  measured  by $v\in \Sigma$, 
	\[T_v(v,v):=\lim_{Vol_{g_v}(B(v))\rightarrow 0}\frac{g_v(P_B,v)}{Vol_{g_v}(B(v))}=\lim_{Vol_{g_v}(B(v))\rightarrow 0}\frac{E_B}{Vol_{g_v}(B(v))},\]
 being $E_B:=g_v(P_B,v)$ the measured energy. 
	\item If  $w$  is $g_v$-orthogonal to $v$ and $g_v$-unit, $T_v(w,v)$ measures the flow of energy per unit of $g_v$-volume in a surface $g_v$-orthogonal to $v$ and $w$  (i.e. some small surface of area $A$ flowing a lapse $\Delta t$),  while $T_v(v,u)$ measures the $w$-momentum density,
	\[T_v(w,v):=\lim_{Vol_{g_v}(B(w))\rightarrow 0}\frac{g_v(P_B,v)}{Vol_{g_v}(B(w))}=\lim_{Vol_{g_v}(A)\rightarrow 0}\frac{1}{A}\left\{\lim_{\Delta t\rightarrow 0}\frac{E_B}{\Delta t}\right\}.\]
	\[  -  T_v(v,w):=\lim_{Vol_{g_v}(B(v))\rightarrow 0}\frac{g_v(P_B,w)}{Vol_{g_v}(B(v))}.\]
	\item If $ z ,w$ are $g_v$-orthogonal to $v$ and $g_v$-unit, $T_v(z,w)$ measures the flow of $w$-momentum per unit of $g_v$-volume in a surface $g_v$-orthogonal to $v$ and $z$,
		 \[  -  T_v(z,w):= \lim_{Vol_{g_v}(B(z))\rightarrow 0}\frac{g_v(P_B,w)}{Vol_{g_v}(B(z))}=\lim_{Vol_{g_v}(A)\rightarrow 0}\frac{1}{A}\left\{\lim_{\Delta t\rightarrow 0}\frac{g_v(P_B,w)}{\Delta t}\right\}.\]
\end{enumerate}

\subsection{Lagrangian viewpoint}\label{s_lagr} In the Lagrangian approach for  Special Relativity,  the background spacetime is assumed to be endowed with a flat metric $\eta$. So, the Lagrangian $\mathcal{L}$ is constructed by using the prescribed $\eta$ and  some matter fields $\phi_\alpha$.
The stress-energy tensor coincides with the {\em canonical}  energy-momentum  tensor associated  with  the Lagrangian, in most cases (the exceptions include theories involving  spin). This canonical  tensor appears as the Noether current  associated with the invariance by spacetime translations (i.e., when $\mathcal{L}(\phi_\alpha,\partial_\mu \phi_\alpha, x^\mu)\equiv \mathcal{L}(\phi_\alpha, \partial_\mu \phi_\alpha)$) , namely\footnote{See for example  \cite{Wald}  (around formula (E.1.36))  or  \cite[\S  32]{LL}.}

 \begin{equation} \label{e_SLagSET}T^{\mu\nu}= \frac{\partial \mathcal{L}}{\partial (\partial_\mu \phi_\alpha)}  \partial^\nu \phi_\alpha -\eta^{\mu\nu} \mathcal{L}.\end{equation}
 In principle, these interpretations would hold unaltered for the case of an affine space with a Lorentz norm,  including the case of  Very Special Relativity.

In General Relativity, however, the Lagrangian formulation  introduces a background Lagrangian independent of matter fields (the Einstein-Hilbert one, eventually with a cosmological constant) and, then, a  matter Lagrangian $\mathcal{L}_{matter}$
which includes a constant of coupling with the background. Then, the  safest way to define the stress-energy is the canonical one obtained as the corresponding action  term $\delta S_{matter}/ \delta  g^{\mu\nu}$  in the Euler-Lagrange equations\footnote{See for example, \cite[\S E.1]{Wald}, \cite[\S 4.3]{Carroll}, \cite[\S 21.2, \S 21.3]{MTW}. 
}, 
\begin{equation} \label{e_LagSET}T_{\mu\nu}=-2 \frac{\delta \mathcal{L}_{matter}}{\delta g^{\mu\nu}}+ g_{\mu\nu} \mathcal{L}_{matter}.\end{equation}
Any tensor obtained in this way will have some advantages to play the role of a stress-energy tensor, because it will be automatically symmetric (in contrast to \eqref{e_SLagSET}) and will have vanishing divergence.

In the Finslerian setting, 
 the variational viewpoint  has been  systematically studied  in a very recent paper  by Hohmann, Pfeifer and Voicu \cite{HPV21}. 
Previously,  the background Lagrangian closest to the Einstein-Hilbert functional in the Finslerian setting had  been studied  in  \cite{PW11, HPV}. Such a functional is obtained as the integral of the Ricci scalar function on the indicatrix of the Lorentz-Finsler metric\footnote{Some arguments which support strongly their choice are (see \cite{HPV0}): (a) the simplest analogous to the vacuum Einstein equation in the Finslerian approach  Ricci$=0$ (proposed by  Rund \cite{Rund},  and satisfied by Finsler pp-waves \cite{FP16})  is not a variational equation, (b) the Ricci scalar functional yields an Euler-Lagrange equation which agrees with Einstein's in the vacuum Lorentz case,  and (c) this Euler-Lagrange equation is the variational completion of the Finslerian   Ricci$=0$.}  $L$.
 Taking into account this background functional, they define the
 energy-momentum scalar function by taking the corresponding variational action  term \cite[formula (84)]{HPV21},
$$\mathfrak{T}= -2\frac{L^3}{|g|}\frac{\delta \mathcal{L}_{matter}}{\delta L}.$$
Notice that, here,  the functional coordinate for the Lagrangian is $L$ and, thus, an (anisotropic) function rather than a 2-tensor is obtained.  However, starting at this function some tensors become useful \cite[formulas (88), (91)]{HPV21}, in particular  a canonically associated (anisotropic   Liouville) 2-tensor  
$$\Theta^\mu_\nu=\frac{\mathfrak{T}}{L} \; \CC^\mu \, \CC_\nu $$
as in Example \ref{ex_Lioville stress energy tensor}.

Notice that, essentially, the information of  these  tensors is codified in $\mathfrak{T}$. Even though such a tensor is justified by the procedure of Gotay-Mardsen in \cite{GM}, some
issues as the following ones might deserve interest for a further discussion:
\ben\item This is not the unique natural possibility to construct an anisotropic 2-tensor starting at $\mathfrak{T}$. For example, an alternative would be  the vertical Hessian\footnote{The multiplication by $L$ is so that taking second vertical derivatives of the $2$-homogeneous $\mathfrak{T}L$ produces a $0$-homogeneous tensor, in the same way that the vertical Hessian of the $2$-homogeneous function $L$ is the $0$-homogeneous fundamental tensor $g$.},
\begin{equation}\label{e_vert_hess}
T_{\mu\nu}=\dot\partial_{\mu,\nu} (\mathfrak{T}L) \equiv 
\frac{\partial^2 (\mathfrak{T}L)}{\partial y^{\mu}\partial y^{\nu}}.
\end{equation}
It is natural to wonder about the choice  closer to the relativistic intuitions about the  stress-energy.
 \item Recently, the Palatini approach has also been studied   for the Finslerian setting \cite{JSV2}. 
There, the dynamic variables are $L$ and the components of an (independent) non-linear connection. Thus, a similar Lagrangian procedure  would lead to a  higher order tensor. In the relativistic setting   this approach supports classical Relativity, as  it recovers both  equations and (in the symmetric case) the Levi-Civita connection. However, the Palatini approach is no longer equivalent in the Finslerian case, as it yields non-equivalent connections and it shows a variety of possibilities for the non-linear connections.
So, it is natural to wonder about the most natural choice  of a Lagrangian-based stress-energy tensor in this setting.
\een

  Finally, let us discuss an example analyzed from the Lagrangian viewpoint in  \cite{HPV0, HPV21}  taking into account also the observers' one  in \S \ref{s3.2}.

\begin{example}\label{ex_kinetic}
  The gravitational field sourced by a kinetic gas has been deeply studied in \cite{HPV0, HPV21}. In the relativistic setting, this  is derived from the Einstein-Vlasov equations in terms of a 1 particle distribution function (1PDF) $\phi(x,\dot x)$ which encodes how many gas particles at a given spacetime point $x$ propagate on
worldlines with normalized 4-velocity $\dot x$.  Specifically,  the stress energy tensor is:
$$
T^{\mu\nu}(x)= \int_{\Sigma_x} \dot{x}^\mu \dot{x}^\nu\phi(x,\dot x) dvol_{g_x},
\qquad \qquad x\in M,$$
being $\Sigma_x$ the indicatrix (future-directed unit vectors of the Lorentz metric)  and $\hbox{dVol}_x$ the volume at each $x$.
In \cite{HPV0}, they propose  to derive the
gravitational field of a kinetic gas directly from the 1PDF without averaging, i.e., taking into account  the full information on the
velocity distribution. This leads to consider the function $\phi: \Sigma \rightarrow \R$, $u\equiv (x,\dot x) \mapsto \phi(u)\geq 0$ as  an   energy-momentum  function  which plays the role of a stress-energy tensor (even though it is a scalar rather than a 2-tensor).  Moreover, the original Lorentz metric is naturally allowed to be Lorentz-Finsler, which permits to obtain more general cosmological models \cite[\S III]{HPV0}. 

Indeed, up to a coupling constant, $\phi$ is regarded directly as the matter source in the Finslerian Einstein-Hilbert equation  (i. e., it is placed at the right-hand side of this equation, \cite[eqn. (7)]{HPV0}). It is worth pointing out:
\bit\item $\phi$ can be reobtained as a Lagrangian energy-momentum by inserting  it directly  as a term in the background Lagrangian \cite[eqn. (75)]{HPV21}. However, the Lagrangian is not natural then, as it depends on the variables of $M$ (recall  \cite[Appendix 3, \S (a)]{HPV21}). 
\item As discussed above, such a function allows one to construct several tensors, in particular the vertical Hessian $\partial^2 \phi/\partial\dot x^\mu \partial\dot x^\nu$ (as in \eqref{e_vert_hess}), which also might play a role to  compare with the relativistic $T^{\mu\nu}(x)$. 
\eit

Anyway, starting at the 1PDF $\phi$,  	another Finslerian interpretations would be possible. In particular,  one can define   
the {\em energy momentum distribution} $\phi(u)u$. Then, given an observer $v\in \Sigma$ and a $g_v$-unit vector, the $w${\em -energy momentum} might be defined as
$$g_v(u,w) \phi(u).$$ 
In particular, when $w=v$ this would be the {\em energy perceived by $v$} and when $w$ is unit and $g_v$-orthogonal to $v$ would be (minus) the {\em momentum in the direction $w$} (compare with the discussion at the end of \S \ref{s3.2}).
So, an alternative stress-energy tensor perceived by each observer $v\in \Sigma$ might be defined as the anisotropic tensor:
\[T_v(w,z)=\int_{\Sigma_{\pi(v)}} g_v(u,w)g_v(u,z)\phi(u)  d\hbox{vol}_{g_v} ,\]
where the integration in $u$ is carried out with the volume form of $(\Sigma_{\pi(v)}, g_v)$, denoted
by $dvol_{g_v}$.
\end{example} 

\section{Divergence of anisotropic vector fields} \label{sec:divergence vectors}

After studying the basic properties of the Finslerian stress-energy tensor $T$, our next aim is to analyze the meaning and significance of the {\em infinitesimal conservation law} $\di(T)=0$. Along this and the next section, we will always consider an anisotropic tensor $T\in\ten_1^1(M_A)$ interpreted as an endomorphism of anisotropic vector fields. $T^\flat\in\ten_2^0(M_A)$ and $T^\sharp\in\ten_2^0(M_A)$ will be defined on vectors and $1$-forms by $T^\flat(X,Y):=g(X,T(Y))$ and $T^\sharp(\theta,\eta):=g^\ast(T^\ast(\theta),\eta)$ resp., where $g^\ast$ is the inverse fundamental tensor and $T^\ast$ is the transpose of $T$. They will have components $\left(T^\flat\right)_{ij}=g_{il}T^l_j=:T_{ij}$ and $\left(T^\sharp\right)^{ij}=T^i_lg^{lj}=:T^{ij}$, and in principle we will not even assume that these are symmetric. We will be assuming that $M$ is orientable an oriented. This is not restrictive: one could always reduce the theory to this case by pulling back all the objects (the fibered manifold $A\rightarrow M$ included) to the oriented double cover of $M$ \cite[Ch. 15]{Lee}.

Let us briefly recall the mathematically precise meaning of the conservation laws in classical General Relativity ($g$, $T$ and $X$ isotropic). One has 
\begin{equation}
	\mathrm{div}(T(X))=\nabla_i(T^i_jX^j)=\nabla_iT^i_jX^j+T^i_j\nabla_iX^j=\di(T)(X)+\mathrm{trace}(T(\nabla X))
	\label{eq:div(T(X))}
\end{equation}
with $\nabla$ the Levi-Civita connection. The first contribution vanishes due to $\di(T)=0$, and there are different situations in which the second one vanishes as well. For instance, if $T^\flat(-,\nabla_{-}X)$ is antisymmetric, then
\begin{equation}
	\mathrm{trace}(T(\nabla X))=T^i_j\nabla_iX^j=g^{il}T_{lj}\nabla_iX^j=\frac{1}{2}g^{il}\left(T_{lj}\nabla_iX^j+T_{ij}\nabla_lX^j\right)=0,
	\label{eq:antisymmetry}
\end{equation}
and if $T^\flat$ is symmetric and $\nabla X^\sharp$ is antisymmetric (equiv., $X$ is a Killing vector field), then also
\begin{equation}
	\mathrm{trace}(T(\nabla X))=g^{il}T_{lj}\nabla_iX^j=\frac{1}{2}T_{lj}\left(g^{li}\nabla_iX^j+g^{ji}\nabla_iX^l\right)=0.
	\label{eq:killing}
\end{equation}
Anyway, whenever $\mathrm{trace}(T(\nabla X))=0$, one can integrate \eqref{eq:div(T(X))} and apply the pseudo-Riemannian divergence theorem to get the {\em integral conservation law}
\begin{equation}
	\int_{\partial D}\imath_{T(X)}(d\mathrm{Vol})=0,
	\label{eq:riemannian conservation law}
\end{equation}
where $\overline{D}$ is a domain of appropriate regularity, $\imath$ is the interior product operator and $d\mathrm{Vol}$ is the metric volume form. In a sense that will be made more precise in \S 5, this is expressing that the total amount of $X$-momentum in a space region only changes along time as much as it flows across the spatial boundary of the region. In this way, there is no ``creation" nor ``destruction" of $X$-momentum in any space region.

Extending the infinitesimal or the integral conservation laws poses, first and foremost, the problem of appropriately defining the divergence of an anisotropic $T$. Observe that a priori it is not clear even how to define the divergence of a vector field $Z$, isotropic or not, as one could consider $\mathrm{trace}(\nabla Z)$ for different anisotropic connections $\nabla$, mainly Chern's and Berwald's. An alternative is to seek for a more geometric, hence {\em unbiased}, definition. For instance, the {\em metric (anisotropic) volume form of $L$},
\begin{equation}
	\vol=\sqrt{\left|\det g_{ab}(x,y)\right|}\dd x^1\wedge...\wedge\dd x^n\in\form_n(M_A)
	\label{eq:metric volume form}
\end{equation}
for $(x^1,...,x^n)$ positively oriented, is well-defined, and when $Z\in\vect(M)$ (i. e., $Z$ is isotropic), so is the Lie derivative
\[
\Lie_Z\colon\ten(M_A)\rightarrow\ten(M_A)
\]
(see \cite[\S 5]{Jav1}). So, by analogy with the classical case, one could think of $\Lie_Z(\vol)$ for defining $\di(Z)$.

It turns out that the unbiased definition, including all $Z\in\vect(M_A)$, is achieved with a modification of this Lie derivative that we will regard as an extension of the classical Lie bracket. We devote the next subsection to the technical mathematical foundations of such an {\em anisotropic Lie bracket}, which needs of a nonlinear connection on $A\rightarrow M$ to be well-defined. All the maps $\ten(M_A)\rightarrow\ten(M_A)$ that will appear in \S \ref{subsec:bracket cartan} will be {\em (anisotropic) tensor derivations} in the sense of \cite[Def. 2.6]{Jav1} and their local nature will be apparent, so we will not explicitly discuss it. For example, the Lie derivative along $Z\in\vect(M)$ is the only tensor derivation such that for $X\in\vect(M)$ and $f\in\fun(A)$, 
\begin{equation}
	\Lie_ZX=\left[Z,X\right],\qquad\Lie_Zf=Z^{\mathrm{c}}(f) := Z^k \frac{\partial f}{\partial x^k}  + y^k \frac{\partial Z^i}{\partial x^k} \frac{\partial f}{\partial y^i}.
	\label{eq:classical lie derivative}
\end{equation}

\subsection{Mathematical formalism of the anisotropic Lie bracket} \label{subsec:bracket cartan}

During this subsection, we fix an arbitrary nonlinear connection given by $\TT A = \HH A \oplus \VV A$ or by the nonlinear covariant derivative $\mathrm{D}$ (keep in mind \eqref{eq:nonlinear connection} and \eqref{eq:nonlinear covariant derivative}), and also an anisotropic vector field $Z\in\vect(M_A)$.

For $X\in\vect(M_A)$, it is very natural to consider the commutator of the horizontal lifts of $Z$ and $X$:
\[
\left[Z^{\mathrm{H}},X^{\mathrm{H}}\right] = \left[Z^j \delta_j,X^k \delta_k\right] = \left(Z^j \delta_jX^i - X^j \delta_jZ^i\right)\delta_i + Z^j X^k \left[\delta_j,\delta_k\right]\in\vect(A).
\]
We recall that $Z^j X^k \left[\delta_j,\delta_k\right]$ is always vertical. Indeed, $\left[\delta_j,\delta_k\right]=\mathcal{R}_{jk}^i \dv_i$, where $\mathcal{R}$ is the curvature tensor of the nonlinear connection (see \cite{JSV2}, where this curvature is regarded as an anisotropic tensor and the homogeneity of the connection is not really required). This means that the horizontal part of $\left[Z^{\mathrm{H}},X^{\mathrm{H}}\right]$ has coordinates $Z^j \delta_jX^i - X^j \delta_jZ^i$, and this corresponds to a globally well-defined $A$-anisotropic vector field:
\begin{equation}
	\lie_Z X := \left(Z^j \delta_jX^i - X^j \delta_jZ^i\right)\partial_i \in\vect(M_A).
	\label{eq:lie bracket vector}
\end{equation}

\begin{defi}\label{d_4.1}
	$\lie_Z X$ is the {\em anisotropic Lie bracket of $Z$ and $X$ with respect to the nonlinear connection $\HH A$}.
\end{defi}

\begin{rem} \label{rem:anisotropic lie bracket}
	The word ``anisotropic" could be omited in the previous definition, in the sense that for $Z,X\in\vect(M_A)$, there is no other Lie bracket, isotropic or not, defined in general. Nonetheless, \eqref{eq:lie bracket vector} makes apparent that when $Z,X\in\vect(M)$ (i. e., when $Z$ and $X$ are isotropic), $\lie_Z X$ coincides with the standard Lie bracket $\left[Z,X\right]$ regardless of the connection.
\end{rem}

\begin{lemma}
	Given a nonlinear connection $\HH A$, $V\in \mathfrak{X}^A(U)$, $f\in{\mathcal F}(A)$ and anisotropic vector fields $X,Z\in \mathfrak{X}(M_A)$, it holds that
	\begin{equation}
		Z^\HH(f)=Z(f(V))-\dot\partial_{\mathrm{D}_ZV}f,
		\label{ZHf}
	\end{equation}
	\begin{equation}
	\left(\lie_ZX\right)_V=\left[Z_V,X_V\right]-\left(\dot\partial_{\mathrm{D}_ZV}X\right)_V+\left(\dot\partial_{\mathrm{D}_XV}Z\right)_V.
	\label{lieXZV}
	\end{equation}
\end{lemma}

\begin{proof}
	Observe that 
	\[
	\begin{split}
	Z(f(V))-\dot\partial_{\mathrm{D}_ZV}f&=Z^i\left(\frac{\partial f}{\partial x^i}(V)+\frac{\partial f}{\partial y^j}(V)\frac{\partial V^j}{\partial x^i}\right) \\
	&\quad-\frac{\partial f}{\partial y^j}(V)Z^k\left(\frac{\partial V^j}{\partial x^k}-N^j_k(V)\right)\\
	&=Z^i\left(\frac{\partial f}{\partial x^i}(V)-\frac{\partial f}{\partial y^j}(V)N^j_i(V)\right) \\
	&=Z^\HH (f),
	\end{split}
	\]
	which concludes \eqref{ZHf}. In particular, $\delta_if(V)=\partial_i(f(V))-\left(\dot\partial_{\mathrm{D}_{\partial_i}V}f\right)(V)$, and using this in \eqref{eq:lie bracket vector}, \eqref{lieXZV} follows.
\end{proof}

We also recall that the {\em torsion} of an $A$-anisotropic connection $\nabla$ \cite[(18)]{Jav1}, \cite[Def. 5]{JSV1} is the anisotropic tensor $\mathrm{Tor}\in\ten_2^1(M_A)$ defined on first on isotropic fields $Z,X\in\vect(M)$ by $\mathrm{Tor}(Z,X) = \nabla_Z X-\nabla_X Z - \left[Z,X\right]$ and then extended by $\fun(A)$-bilinearity. Therefore, it can be regarded as and $\fun(A)$-bilinear map $\mathrm{Tor}\colon\vect(M_A)\times\vect(M_A)\rightarrow\vect(M_A)$ and it has coordinates 
\begin{equation}
	\mathrm{Tor}_{jk}^i=\varGamma_{jk}^i-\varGamma_{kj}^i,
	\label{eq:tor}
\end{equation}
where the $\varGamma_{jk}^i$'s are the Christoffel symbols of $\nabla$.\footnote{This is not to be mistaken by the torsion of the nonlinear connection $\HH A$, which would have coordinates $N_{j\,\cdot k}^i - N_{k\,\cdot j}^i$ (even though this can be seen as a particular case of the torsion of some $\nabla$ and hence it is also denoted by $\mathrm{Tor}$ in \cite{JSV2}).}

\begin{thm} \label{th:lie bracket} Let a nonlinear connection $\TT A=\HH A \oplus \VV A$ and an anisotropic vector field $Z\in\vect(M_A)$ be fixed. 
	
\noindent (A) If $\nabla$ is any $A$-anisotropic connection whose underlying nonlinear connection is $\HH A$, then for any $X\in\vect(M_A)$,
\begin{equation}
\mathrm{Tor}(Z,X) = \nabla_Z X -\nabla_X Z - \lie_{Z}X
\label{eq:lie bracket in terms of nabla}
\end{equation}
(where $\mathrm{Tor}$ is the torsion of $\nabla$).
	
\noindent (B) By imposing the Leibniz rule with respect to tensor products and the commutativity with contractions, the map $X\mapsto\lie_{Z}X$ extends unequivocally to an (anisotropic) tensor derivation $\lie_{Z}\colon\ten^r_s(M_A)\rightarrow\ten^r_s(M_A)$ given by 
\begin{equation}
	\begin{split}
	\lie_{Z}T(\theta^1,...,\theta^r,X_1,...,X_s)&=Z^\mathrm{H}(T(\theta^1,...,\theta^r,X_1,...,X_s))
	\\
	&\quad-\sum_{\mu=1}^r T(\theta^1,...,\lie_{Z}\theta^\mu,...,\theta^r,X_1,...,X_s)
	\\
	&\quad-\sum_{\nu=1}^s T(\theta^1,...,\theta^r,X_1,...,\lie_{Z}X_{\nu},...,X_s)
	\end{split}
	\label{eq:lie bracket tensor}
\end{equation}
for $\theta^\mu \in \varOmega_1(M)$ and $X_\nu\in\mathfrak{X}(M)$. In coordinates, if 
\[
T=T^{i_1,...,i_r}_{j_1,...,j_s}(x,y)\partial_{i_1}\otimes...\otimes\partial_{i_r}\otimes\dd x^{j_1}\otimes...\otimes\dd x^{j_s},
\]
then 
\begin{equation}
\left(\lie_{Z}T\right)^{i_1,...,i_r}_{j_1,...,j_s}=Z^k\frac{\delta T^{i_1,...,i_r}_{j_1,...,j_s}}{\delta x^k}-\sum_{\mu=1}^r\frac{\delta Z^{i_\mu}}{\delta x^k}T^{i_1,...,k,...,i_r}_{j_1,...,j_s}+\sum_{\nu=1}^s\frac{\delta Z^{k}}{\delta x^{j_\nu}}T^{i_1,...,i_r}_{j_1,...,k,...,j_s}.
\label{eq:lie bracket coordinates}
\end{equation}
	
\noindent (C) The map 
\[
\Lie_Z^\HH := \lie_{Z} - \dv_{\lie_{Z}\CC}\colon\ten(M_A)\rightarrow\ten(M_A)
\]
is also a tensor derivation. When $Z\in\vect(M)$,
\begin{equation}
	\Lie_Z^\HH T = \Lie_{Z}T
	\label{eq:consistency lie derivative}
\end{equation}
for all $T\in\ten(M_A)$, where $\Lie_{Z}$ is the Lie derivative \eqref{eq:classical lie derivative}, regardless of the nonlinear connection.

\noindent (D)  Given $V\in \mathfrak{X}^A(U)$ and $\omega\in\varOmega_n(M_A)$ ($n=\dim M$), it holds that
\begin{equation}\label{lieV}
\left(\lie_{Z} \omega\right)_V =\Lie_{Z_V} (\omega_V)-\dot\partial_{\mathrm{D}_{Z}V}\omega-\mathrm{trace}(\dot\partial _{\mathrm{D}V}Z)\omega.
\end{equation}
\end{thm}

\begin{proof}
	(A) It is straightforward to compute that the right hand side of \eqref{eq:lie bracket in terms of nabla} is ${\mathcal F}(A)$-multilinear. Moreover, the identity is trivial on isotropic vector fields $X,Z\in\mathfrak{X}(M)$, as $\lie_ZX=\left[X,Z\right]$ in this case, which concludes.
	
	(B) Given $f\in\ten^0_0(M_A)=\fun(A)$, for $X\in\ten^1_0(M_A)=\vect(M_A)$ it follows from \eqref{eq:lie bracket vector} that
	\[
	\lie_{Z}(f X)=Z^\mathrm{H}(f) X + f\,\lie_{Z}X.
	\]
	Thus, in order to respect the Leibniz rule, the only possibility is to define 
	\begin{equation}
	\lie_{Z}f=Z^\mathrm{H}(f)=Z^k \frac{\delta f}{\delta x^k}.
	\label{eq:lie bracket function}
	\end{equation}
	Now, given $\theta\in\ten^0_1(M_A)=\form_1(M_A)$, in order to respect again the Leibniz rule and the commutativity with contractions, the only possibility is to define $\lie_{Z} \theta$ on every $X\in\vect(M_A)$ by 
	\begin{equation}
	\left(\lie_{Z}\theta\right)(X)=Z^\mathrm{H}(\theta(X))-\theta(\lie_{Z} X)=\left(Z^k \frac{\delta \theta_j}{\delta x^k}+\frac{\delta Z^k}{\delta x^j} \theta_k\right)X^j. 
	\label{eq:lie bracket form}
	\end{equation}
	\eqref{eq:lie bracket function}, \eqref{eq:lie bracket vector} and \eqref{eq:lie bracket form} make apparent that $\lie_Z$ is already local on functions, vector fields and $1$-forms, and they allow to compute 
	\begin{equation}
	\lie_Z(\partial_i)=-\frac{\delta Z^k}{\delta x^i}\partial_k,\qquad\lie_Z(\dd x^j)=\frac{\delta Z^j}{\delta x^k}\dd x^k.
	\label{eq:lie partial and dx}
	\end{equation}
	Finally, given $T\in\mathcal{T}^r_s(M_A)$, one is led to define $\lie_Z T$ by \eqref{eq:lie bracket tensor}. Clearly, this indeed provides a tensor derivation and \eqref{eq:lie bracket coordinates} follows from the evaluation of \eqref{eq:lie bracket tensor} at $(\dd x^{i_1},...,\dd x^{i_r},$ $\partial_{j_1},...,\partial_{j_s})$
	together with \eqref{eq:lie bracket function} and \eqref{eq:lie partial and dx}. 
	
	(C) $\dv_X\colon\ten(M_A)\rightarrow\ten(M_A)$ is a tensor derivation for any $X\in\vect(M_A)$, in particular for 
	\begin{equation}
	X = \lie_{Z}\CC = \left(Z^j \delta_jy^i - y^j \delta_jZ^i\right)\partial_i = -\left(Z^j N_j^i + y^j \delta_jZ^i\right)\partial_i 
	\label{eq:lie bracket can}
	\end{equation}
	(see \eqref{eq:lie bracket vector}). Thus, the difference $\Lie_Z^\HH = \lie_{Z} - \dv_{\lie_{Z}\CC}$ is again a derivation. As for the last assertion, where $Z\in\vect(M)$, we are going to use \cite[Prop. 2.7]{Jav1}. For $X\in\vect(M)$, we have
	\begin{equation}
	\Lie_Z^\HH X = \lie_Z X = \left[Z,X\right] = \Lie_Z X
	\label{eq:consistency lie bracket}
	\end{equation}
	(recall Rem. \ref{rem:anisotropic lie bracket}). For $f\in\fun(A)$, we have
	\[
	\begin{split}
	\Lie_Z^\HH f = \lie_{Z}f - \dv_{\lie_{Z}\CC}f &= Z^j \delta_jf + \left(Z^j N_j^i + y^j \delta_jZ^i\right)\dv_if \\
	&= Z^j \left(\partial_jf - N_j^i \dv_ if\right) + \left(Z^j N_j^i + y^j \delta_jZ^i\right)\dv_if \\
	&= Z^j \partial_jf + y^j \delta_jZ^i \dv_if \\
	&= \Lie_Z f
	\end{split}
	\] 
	(see \eqref{eq:lie bracket function}, \eqref{eq:lie bracket can}, \eqref{eq:nonlinear connection} and \eqref{eq:classical lie derivative}). As $\Lie_Z^\HH$ and $\Lie_Z$ act the same on isotropic vector field and anisotropic functions, they are equal.
	
	(D) Observe that for $X\in \mathfrak{X}(M)$, the term $\dot\partial_{\mathrm{D}_ZV}X$  vanishes in \eqref{lieXZV}. Moreover,  if $Z\in \mathfrak{X}(M_A)$ and $f\in{\mathcal F}(A)$, then $Z^\HH(f)_V=Z_V(f(V))-\left(\dot\partial_{\mathrm{D}_ZV}f\right)(V)$. Given a local reference frame $E_1,...,E_n\in \mathfrak{X}(U)$, and taking into account the last two identities and the definitions of $\lie$ and $\Lie$, it follows that
	\[
	\begin{split}
	\left(\lie_{Z} \omega\right)_V(E_1,...,E_n) -\Lie_{Z_V} (\omega_V)(E_1,...,E_n)&=-\dot\partial_{D_ZV}\omega(E_1,...,E_n)\\
	&\quad-\sum_{i=1}^n \omega(E_1,..., \dot\partial_{D_{E_i} V}Z,..., E_n).
	\end{split}
	\]
	As $\omega(E_1,..., \dot\partial_{D_{E_i} V}Z,..., E_n)= E_i^*(\dot\partial_{D_{E_i} V}Z) \omega_V(E_1,...,E_n)$, \eqref{lieV} follows.
\end{proof}

\begin{defi}\label{d_4.5}
	The tensor derivation $\lie_Z\colon\ten(M_A)\rightarrow\ten(M_A)$ defined in Th. \ref{th:lie bracket} (B) is the {\em (anisotropic) Lie bracket with $Z$}, while $\Lie^\HH_Z\colon\ten(M_A)\rightarrow\ten(M_A)$ is the {\em (anisotropic) Lie derivative along $Z$}, both of them {\em with respect to the connection $\HH A$}.
\end{defi}

\begin{rem} [Anisotropic Lie bracket and Lie derivative]
	The derivation $\Lie_Z^\HH$ defined in Th. \ref{th:lie bracket} (C) would be the {\em Lie derivative along $Z$ with respect to $\HH A$}. Analogously to the discussion of Rem. \ref{rem:anisotropic lie bracket}, what makes this name consistent is \eqref{eq:consistency lie derivative}: whenever the Lie derivative along $Z$ was already defined, $\Lie^\HH_Z$ coincides with it. Even though the Lie bracket and the Lie derivative are equal in the classical regime, it is heuristically useful to regard $\lie$ as the anisotropic generalization of the former and $\Lie^\HH$ as that of the latter, in order to distinguish them. It is actually $\lie$, and not $\Lie$, which will be relevant for the definition of divergence. The reason is that the former, as we will see below, has a clear geometric interpretation in terms of flows, while the latter would just add the term $\dv_{\lie_{Z}\CC}$ to that interpretation. Moreover, Th. \ref{th:lie bracket} (D) actually corresponds to a Cartan formula for $\Lie_Z$ whose full development we postpone for a future work. Thus, $\Lie_Z(\vol)=\Lie_Z^\HH(\vol)$ can be regarded as an initial guess for the divergence of $Z$, but we will not employ $\Lie^\HH$ from now on.
\end{rem}

Let us observe that given a diffeomorphism $\psi_t\colon M\rightarrow M$ that is the flow of an isotropic vector field $Z$, we can define the pullback $\psi_t^*(\omega)$ of an anisotropic differential form $\omega\in  \form_s(M_A)$ as the anisotropic form given by $\psi_t^*(\omega)_v(u_1,...,u_s):=\omega_{P_t(v)}(\dd\psi_t(u_1),...,\dd\psi_t(u_s))$, where $P_t(v)$ is the $\HH A$-parallel transport of $v$ along the integral curve of $Z$ and $u_1,..., u_s\in \TT_{\pi(v)}M$.

\begin{prop}\label{Lieflowformula}
	If $Z\in \mathfrak{X}(M)$ and $\omega\in \form_s(M_A)$, then 
	\begin{equation}\label{lieH}
	\lie_Z\omega=\lim_{t\rightarrow 0}\frac{\psi^*_t(\omega)-\omega}{t},
	\end{equation}
	where $\psi_t$ is the (possibly local) flow of $Z$.  
\end{prop}

\begin{proof}
	Observe that $\psi_t^*(\omega)_v$ can be obtained as $\psi_t^*(\omega_V)$ with $V$ an extension of $v$ such that $\mathrm{D}_ZV=0$. Then \eqref{lieV} and the classical formula for the Lie derivative in terms of the flow imply \eqref{lieH}.
\end{proof}

\begin{rem}
	Even though, for convenience, we stated the previous geometrical interpretation for an $s$-form $\omega$, it should be clear that it holds true for any $r$-contravariant $s$-covariant $A$-anisotropic tensor.
\end{rem}

\subsection{Lie Bracket definition of divergence} \label{s_4.2}
Finally, in this and the next subsections a pseudo-Finsler metric $L$ defined on $A$ is fixed again. In its presence, and in view of the Riemannian case and Prop. \ref{Lieflowformula}, the most natural way of defining the divergence of an anisotropic vector field $Z$ is by $\lie_Z(\vol)$. Here there is a canonical choice for $\HH A$: the metric nonlinear connection of $L$. The definition obtained this way is unbiased, in that one does not choose any anisotropic connection {\em a priori}. Notwithstanding, it will turn out to be most conveniently expressed in terms of the Chern connection.  

\begin{defi}\label{d_4.9} For $Z\in\mathfrak{X}(M_A)$, its \emph{divergence with respect to the pseudo-Finsler metric $L$} is the anisotropic function $\di(Z)\in\fun(A)$ defined by 
	\[
	\lie_Z(\vol) =: \di(Z)\vol,
	\]
where $\HH A$ and $\vol$ are, resp., the metric nonlinear connection \eqref{eq:metric connection} and the metric volume form	\eqref{eq:metric volume form} of $L$.
\end{defi}

\begin{rem}
	Even though we will keep assuming it for simplicity, the hypothesis of $M$ being orientable is not really needed for this definition. As in pseudo-Riemannian geometry, on small enough open sets $U\subseteq M$ it is always possible to choose an orientation, define $\vol_U\in\form_n(M_A)$ with respect to it and put $\left.\di(Z)\right|_{A\cap\TT U}\vol_U := \lie_Z(\vol_U)$. The different definitions will be coherent because when the orientation changes, $\vol_U$ changes to $-\vol_U$ and 
	\[
	\lie_Z(-\vol_U) = -\lie_Z(\vol_U) = \left.-\di(Z)\right|_{A\cap\TT U}\vol_U = \di(Z)_{A\cap\TT U}\left(-\vol_U\right).
	\]
	In particular, when $M$ is orientable, $\di(Z)$ is independent of the orientation choice.
\end{rem}

\begin{prop} \label{prop:divergence vector characterization} Let $L$ be a fixed pseudo-Finsler metric defined on $A$, and let $Z\in\mathfrak{X}(M_A)$. If $\nabla$ is any symmetric $A$-anisotropic connection such that its underlying nonlinear connection is the metric one and $\nabla_Z(\vol)=0$, then
	\begin{equation}
		\mathrm{div}(Z)=\mathrm{trace}(\nabla Z),
		\label{eq:divergence chern}
	\end{equation}	
or in coordinates,
	\begin{equation}
		\mathrm{div}(Z)=\frac{\delta Z^i}{\delta x^i}+\varGamma_{ik}^iZ^k
		\label{eq:divergence vector coordinates}
	\end{equation}
This, in particular, is true for the (Levi-Civita)--Chern anisotropic connection of $L$, so one can take the Christoffel symbols to be those of \eqref{eq:chern}.
\end{prop}

\begin{proof}
	One expresses the $Z$-Lie bracket of the volume form in terms of the anisotropic connection, analogously to the isotropic case. From \eqref{eq:metric volume form} and the fact that $\mathfrak{l}_Z^\mathrm{H}$ is a tensor derivation, we obtain
	\[
	\begin{split}
	\mathrm{div}(Z)\sqrt{\left|\det g_{ab}\right|} &= \mathrm{div}(Z)d\mathrm{Vol}(\partial_1,...,\partial_n) \\
	&= \mathfrak{l}_Z^\mathrm{H}(d\mathrm{Vol})(\partial_1,...,\partial_n)\\
	&=\mathfrak{l}_Z^\mathrm{H}(d\mathrm{Vol}(\partial_1,...,\partial_n))-\sum_{i=1}^nd\mathrm{Vol}(\partial_1,...,\mathfrak{l}_Z^\mathrm{H}\partial_i,...,\partial_n). 
	\end{split}
	\]
	\eqref{eq:lie bracket function} and the fact that $\HH A$ is the underlying nonlinear connection of $\nabla$ give
	\[
	\mathfrak{l}_Z^\mathrm{H}(d\mathrm{Vol}(\partial_1,...,\partial_n)) = Z^\HH(d\mathrm{Vol}(\partial_1,...,\partial_n)) = \nabla_Z(d\mathrm{Vol}(\partial_1,...,\partial_n)).
	\]
	\eqref{eq:lie bracket in terms of nabla} and  $\mathrm{Tor} = 0$ 
	\[
	d\mathrm{Vol}(\partial_1,...,\mathfrak{l}_Z^\mathrm{H}\partial_i,...,\partial_n) = d\mathrm{Vol}(\partial_1,...,\nabla_Z\partial_i,...,\partial_n) - d\mathrm{Vol}(\partial_1,...,\nabla_{\partial_i}Z,...,\partial_n).
	\]
	From these and $\nabla_Z(\vol)=0$, 
	\begin{equation}
		\begin{split}
		\mathrm{div}(Z)\sqrt{\left|\det g_{ab}\right|} &= \nabla_Z(d\mathrm{Vol}(\partial_1,...,\partial_n)) - \sum_{i=1}^nd\mathrm{Vol}(\partial_1,...,\nabla_Z\partial_i,...,\partial_n)\\
		&\quad+\sum_{i=1}^nd\mathrm{Vol}(\partial_1,...,\nabla_{\partial_i}Z,...,\partial_n) \\
		&=\nabla_Z(\vol)(\partial_1,...,\partial_n) + \sum_{i=1}^nd\mathrm{Vol}(\partial_1,...,\nabla_{\partial_i}Z,...,\partial_n) \\
		&=\sum_{i=1}^nd\mathrm{Vol}(\partial_1,...,\nabla_{\partial_i}Z,...,\partial_n) \\
		&=\mathrm{trace}(\nabla Z)\sqrt{\left|\det g_{ab}\right|},
		\label{eq:divergence in terms of nabla}
		\end{split}
	\end{equation}
	where the last equality is reasoned analogously as in the proof of \eqref{lieV}.
	
	For the Chern connection, it can be checked that $\nabla(d\mathrm{Vol})=0$ by considering a parallel orthonormal basis with respect to a parallel observer $V$ along the integral curves of any vector field. The coordinate expression of $\mathrm{trace}(\nabla Z)$ in this case concludes \eqref{eq:divergence vector coordinates}.
\end{proof}

\subsection{Divergence theorem and boundary term representations} \label{sec:divergence tensor}

Our Lie bracket derivation allows us to obtain a statement of the Finslerian divergence theorem that subsumes both Rund's \cite[(3.17)]{Rund} and Minguzzi's \cite[Th. 2]{Min17}. This way, it does not need of computations in coordinates from the beginning nor of the ``pullback metric" ($g_V$ in our notation). Naturally, our statement does not include Shen's \cite[Th. 2.4.2]{Shen}, as this one is an independent generalization of the Riemannian theorem not dealing with anisotropic differential forms nor vector fields.

\begin{lemma}\label{verVol}
	For $X\in\vect(M_A)$, the vertical derivative of $d\mathrm{Vol}$ is given by
	\begin{equation}\label{verVoleq}
	\dot\partial _X (d\mathrm{Vol})= C^\mathfrak{m}(X) d\mathrm{Vol},
	\end{equation}
	where $C^\mathfrak{m}$ is the mean Cartan tensor of $L$ (see \eqref{eq:mean cartan}).
\end{lemma}

\begin{proof}
	Let $E_1(t)$, ..., $E_n(t)$ be a positively oriented $g_{v+tX}$-orthonormal basis for every $t\in [0,\varepsilon]$ for a certain $\varepsilon>0$. 
	Then  $d\mathrm{Vol}_{v+tX}(E_1(t),...,E_n(t))=1$ for all $t\in [0,\varepsilon]$. This implies that
	\[
	\dot\partial_X (d\mathrm{Vol})_{v}(E_1(0),...,E_n(0))+ \sum_{i=1}^nd\mathrm{Vol}_{v}(E_1(0),...,\dot E_i(0),...,E_n(0))=0.
	\]
	Moreover, as $g_{v+tX}(E_i(t),E_i(t))=\pm 1$, 
	\[
	2C_v(E_i(0),E_i(0),X)+2g_v(\dot E_i(0),E_i(0))=0.
	\] 
	Using this relation above, we conclude \eqref{verVoleq}.
\end{proof}

In the present article, by a {\em domain} $\overline{D}$ we understand a nonempty connected set which coincides with the closure of its interior $D$; then its boundary is $\partial\overline{D}=\partial D$. Physically, it is very important to include examples in which different parts of $\partial D$ have different causal characters, and this tipically leads to the boundary not being totally smooth. Hence, we will make a weaker regularity assumption that still allows one to apply Stokes' theorem on $\overline{D}$. A subset of $M$ has {\em $0$ $m$-dimensional measure} if its intersection with any embedded $m$-dimensional submanifold $\sigma\subseteq M$ is of $0$ measure in the smooth manifold $\sigma$. Finally, the interior product of an $s$-form $\omega$ with a vector field $X$ will be
\[
\imath_X\omega:=\omega(X,-,...,-).
\]

\begin{thm} \label{th:divergence theorem} Let $L$ be a fixed pseudo-Finsler metric defined on $A$. If
	\begin{enumerate} [label=(\roman*)]
		\item $Z\in\mathfrak{X}(M_A)$ is an anisotropic vector field, 
		\item $V\in\mathfrak{X}^A(U)$ is an $A$-admissible field with $U\subseteq M$ open, and
		\item $\overline{D}\subseteq U$ is a domain with $\partial D$ smooth up to subset of $0$ $\left(n-1\right)$-dimensional measure on $M$ and $\mathrm{Supp}(Z_V)\cap\overline{D}$ compact,
	\end{enumerate} 
	then
	\begin{equation}
	\begin{split}
	&\quad\int_D\mathrm{div}(Z)_V d\mathrm{Vol}_V+\int_D\left\{C^\mathfrak{m}(\mathrm{D}_ZV)+\mathrm{trace}(\dot\partial _{\mathrm{D}V}Z)\right\}
	d\mathrm{Vol}_V \\
	&=\int_{\partial D}\imath_{Z_{V}}(d\mathrm{Vol}_V),
	\end{split}
	\label{eq:divergence theorem}
	\end{equation}
	where $C^\mathfrak{m}$ is the mean Cartan tensor and $\mathrm{D}V$ is computed with the metric nonlinear connection \eqref{eq:metric connection}.
\end{thm}

\begin{proof}
	The idea is to apply Stokes' theorem to 
	$\Lie_{Z_V}(d\mathrm{Vol}_V)$. But taking into account \eqref{lieV} and Lem. \ref{verVol}, it follows that
	\[
	\Lie_{Z_V}(d\mathrm{Vol}_V)=\lie_Z(d\mathrm{Vol})_V+\left\{C^\mathfrak{m}(\mathrm{D}_ZV)+ {\rm trace} (\dot\partial _{\mathrm{D}V}Z)\right\} d\mathrm{Vol}_V, 
	\]
	concluding \eqref{eq:divergence theorem}.
\end{proof} 

\begin{rem}[Riemannian and Finslerian unit normals] \label{rem:normals} Let $i\colon \Gamma\hookrightarrow M$ be the inclusion of a smooth open subset $\Gamma\subseteq\partial D$.
	\noindent \begin{enumerate} [label=(\roman*)]
		\item Even though we do not use the pseudo-Riemannian metric $g_V$ to derive Th. \ref{th:divergence theorem}, from our physical viewpoint it is natural to use it to re-express the boundary term.
		If $\Gamma$ is non-$g_V$-lightlike, then for a $g_V$-normal field $\widehat{N}_V$ and a transverse field $X$ along $i$, the form
		\begin{equation}
		\qquad\quad d\sigma_V:=\mathrm{sgn}(g_V(\widehat{N}_V,\widehat{N}_V))\frac{\sqrt{\left|g_V(\widehat{N}_V,\widehat{N}_V)\right|}}{g_V(\widehat{N}_V,X)}i^\ast(\imath_X(\vol_V))\in\varOmega_{n-1}(\Gamma)
		\label{eq:riemannian hypersurface form}
		\end{equation}
		is nonvanishing and independent of $X$. In particular, 
		\[
		d\sigma_V=\frac{1}{\sqrt{\left|g_V(\widehat{N}_V,\widehat{N}_V)\right|}}i^\ast(\imath_{\widehat{N}_V}(\vol_V))
		\]
		is independent of the scale of $\widehat{N}_V$, which we will always assume to be $g_V$-unitary and $D$-salient, so
		\[
		d\sigma_V=i^\ast(\imath_{\widehat{N}_V}(\vol_V))
		\]
		coincides with the hypersurface $g_V$-volume form of $\Gamma$. Taking into account that $i^\ast(\imath_{Z_{V}}(d\mathrm{Vol}_V))$ vanishes wherever $Z_{V}$ is tangent to $\Gamma$ and that $g_V(\widehat{N}_V,\widehat{N}_V)=\pm 1$, \eqref{eq:riemannian hypersurface form} allows us to represent
		and the right hand side of \eqref{eq:divergence theorem} as 
		\begin{equation}
		\int_{\Gamma}\imath_{Z_{V}}(d\mathrm{Vol}_V)=\int_{\Gamma}g_V(\widehat{N}_V,\widehat{N}_V)g_V(\widehat{N}_V,Z_V)d\sigma_V.
		\label{eq:boundary term riemannian}
		\end{equation}
		In fact, this is how Rund's divergence theorem follows from Th. \ref{th:divergence theorem}.
		
		\item There is another way that one can try to represent the boundary term. Namely, assume that there exists a smooth $\xi\colon p\in\Gamma\rightarrow\xi_p\in A\cap\TT_{p}M$ with $\TT_p\Gamma=\mathrm{Ker}\,g_{\xi_p}(\xi_p,-)$ and $L(\xi_p)=\pm 1$ (in the Lorentz-Finsler case, it will necessarily be $L(\xi)=1$). This is called a {\em Finslerian unit normal along $\Gamma$}. Analogously as in (i), one can put
		\[
		d\Sigma^\xi_V:=L(\xi)\frac{1}{g_\xi(\xi,X)}i^\ast(\imath_X(\vol_V))=i^\ast(\imath_{\xi}(\vol_V)),
		\]
		\begin{equation}
		\int_{\Gamma}\imath_{Z_{V}}(d\mathrm{Vol}_V)=\int_{\Gamma} \epsilon_\xi L(\xi)g_\xi(\xi,Z_V)d\Sigma^\xi_V;
		\label{eq:boundary term finslerian}
		\end{equation}
		here, due to the possible orientation difference between both sides, 
		\[
		\epsilon_\xi=\begin{cases}
		1, & \text{where $\xi$ is $D$-salient},\\
		-1 & \text{where $\xi$ is $D$-entering.}
		\end{cases}
		\]
		In fact, this is how Minguzzi deduces his divergence theorem \cite[Th. 2]{Min17}. Note, however, that he does it under the hypothesis of vanishing mean Cartan tensor ($C^\mathfrak{m}=0$), which implies that $d\Sigma^\xi_V$ is independent of $V$. As we do not require this, Th. \ref{th:divergence theorem} is more general statement than Minguzzi's. 
		
		\item The Finslerian unit normal presents some issues in the general case, as we are not taking $A=\TT M\setminus 0$. 
		In our physical interpretation, with $L$ Lorentz-Finsler, $A$ consists of timelike vectors, so asking for a Finslerian unit normal is only reasonable when $\Gamma$ is {\em $L$-spacelike}, that is, $\TT_p\Gamma\cap\left(A\cap\partial A\right)=\emptyset$ for $p\in\Gamma$. In such a case, the strong concavity of the indicatrix $\left\{v\in A_p\colon\,L(v)=1\right\}$ guarantees the existence and uniqueness of $\xi$: one defines $\xi_p$ to be the unique vector such that $\TT_p\Gamma+\xi_p$ and the indicatrix are tangent at $\xi_p$.
		
		\item Of course, if $L$ comes from a pseudo-Riemannian metric on $M$, then $\xi=\epsilon_\xi\widehat{N}_V=\epsilon_\xi \widehat{N}$ and $d\Sigma^\xi_V=\epsilon_\xi d\sigma_V=\epsilon_\xi d\sigma$. 
		
		\item It should be clear from this discussion that the form that one integrates on the right hand side of \eqref{eq:divergence theorem} is always the same and that the only difference between Rund's and Minguzzi's divergence theorems is how each of them represents it. Notwithstanding, this is an important difference, for the boundary terms \eqref{eq:boundary term riemannian} and \eqref{eq:boundary term finslerian} could potentially have different physical interpretations. 
	\end{enumerate}
\end{rem} 

\section{Divergence of anisotropic tensor fields}\label{s5}

Our developments of the previous section will allow us to obtain integral Finslerian conservation laws for a tensor $T$ with $\di(T)=0$. We obtain one for each $V\in\mathfrak{X}^A(U)$ satisfying certain hypotheses. Physically, $T$ can be interpreted as an anisotropic stress-energy tensor and $V$ as an observer field. We will also revisit two of the main examples with a clearer physical interpretation: Special Relativity and the conservation of the ``total energy of the universe". In order to do all this, let us see how the Chern connection enters the Finslerian definition of $\di(T)$. 

\subsection{Definition of divergence with the Chern connection}\label{d_5.1}

Prop. \ref{prop:divergence vector characterization} motivates the most natural definition of divergence of $T\in\ten_1^1(M_A)$. Namely, by analogy with the classical case, we shall require \eqref{eq:div(T(X))} to hold for any anisotropic vector field $X\in\vect(M_A)$. This makes the Chern connection appear now: it is the only Finslerian connection $\nabla$ for which one can assure that \eqref{eq:divergence chern} holds independently of $Z:=T(X)$. We shall also explore the conditions under which the term $\mathrm{trace}(\nabla Z)$ vanishes in the general Finslerian setting.

\begin{prop} \label{prop:div(T(X))}
	Let $L$ be a fixed pseudo-Finsler metric defined on $A$ with metric nonlinear connection $\HH A$ and Chern anisotropic connection $\nabla$. Also, let $S\in\ten_2^0(M_A)$ be symmetric, $v\in A$, $T\in\ten_1^1(M_A)$ and $X\in\vect(M_A)$.
	
	\noindent (A) The following are equivalent. 
	\begin{enumerate} [label=(A\roman*)]
		\item $S_v(-,\nabla^v_{-} X)$ is antisymmetric.
		\item $\nabla^v X$ is anti-self-adjoint with respect to $S_v$, that is, $S_v(\nabla^v_{-} X,-) = -S_v(-,\nabla^v_{-} X)$.
		\item $\left(\lie_XS\right)_v= \nabla^v_XS$.
	\end{enumerate}
	
	\noindent (B) One has
	\[
	\di(T(X))-\mathrm{trace}(T(\nabla X))=\mathbf{C}^1_2(\nabla T)(X), 
	\]
	where $\mathbf{C}^1_2$ is the operator that contracts the contravariant index with the covariant one introduced by $\nabla$.
	
	\noindent (C) One has $\mathrm{trace}(T(\nabla X))(v) = 0$ assuming any of the following conditions. 
	\begin{enumerate} [label=(C\roman*)]
		\item $T^\flat_v(-,\nabla^v_{-} X)$ is antisymmetric.
		\item $T^\flat_v$ is symmetric and $\left(\lie_Xg\right)_v= 0$.
	\end{enumerate}
\end{prop}

\begin{proof}
	For (A), take $Y,W\in\vect(M)$. The antisymmetry of $S_v(-,\nabla^v_{-} X)$ reads
	\[
	S_v(\nabla^v_Y X,W)=S_v(W,\nabla^v_Y X)=-S_v(Y,\nabla^v_W X),
	\]
	which is exactly the anti-self-adjointness of $\nabla^v X$ with respect to $S_v$. Besides, \eqref{eq:lie bracket function} and \eqref{eq:lie bracket in terms of nabla} together with $\mathrm{Tor}=0$ for the Chern connection give
	\begin{equation}
	\begin{split}
	&\quad\lie_XS(Y,W) \\
	&=X^\HH(S(Y,W))-S(\lie_XY,W)-S(Y,\lie_XW) \\
	&=X^\HH(S(Y,W))-S(\nabla_XY-\nabla_YX,W)-S(Y,\nabla_XW-\nabla_WX) \\
	&=\nabla_X S(Y,W)+S(\nabla_YX,W)+S(Y,\nabla_WX),
	\end{split}
	\label{eq:lie S in terms of nabla}
	\end{equation}
	which shows that $\left(\lie_X S\right)_v = \nabla^v_X S$ also is equivalent to the anti-self-adjointness.
	
	For (B), all the computations in \eqref{eq:div(T(X))} hold formally the same in the general Finslerian case due to Prop. \ref{prop:divergence vector characterization}. 
	
	As for the vanishing of $\mathrm{trace}(T(\nabla X))(v)$, it follows from (Ci) by the same computations as in \eqref{eq:antisymmetry}. Indeed, the antisymmetry can be expressed as
	\[
	T_{lj}(v)\nabla_iX^j(v)+T_{ij}(v)\nabla_lX^j(v)=0.
	\] 
	It also follows from (Cii) by \eqref{eq:killing}. Indeed, $\left(\lie_X g\right)_v=0$ is equivalent to $\nabla^v X$ being anti-self-adjoint with respect to $g_v$, and this can be expressed as 
	\[
	g^{li}(v)\nabla_iX^j(v)+g^{ji}(v)\nabla_iX^l(v)=0.
	\]
\end{proof}

\begin{rem}[$\lie_X g$ and Finslerian Killing fields] \label{rem:killing} 
	In classical Relativity ($g$, $T$ and $X$ isotropic), the second condition in (C ii) above would read $\left(\Lie_Xg\right)_{\pi(v)}=0$, and $\Lie_Xg=0$ would be equivalent to $X$ being a Killing vector field. In the general case, $X$ being Killing can be defined by the conditions $X\in\vect(M)$ and $\Lie_XL=0$ \cite[\S 5]{Jav1}, but (using Th. \ref{th:lie bracket} (C), the facts that $\dv \CC=\mathrm{Id}$ and $C(\CC,-,-)=0$, and also \eqref{eq:lie S in terms of nabla})
	\[
	\begin{split}
	\Lie_XL&=\Lie_X(g(\CC,\CC)) \\
	&=\Lie_Xg(\CC,\CC)+2g(\Lie_X\CC,\CC) \\
	&=\left(\lie_{X}g - \dv_{\lie_{X}\CC}g\right)(\CC,\CC)+2g(\lie_{X}\CC - \dv_{\lie_{X}\CC}\CC,\CC) \\
	&=\lie_{X}g(\CC,\CC) - 2C(\CC,\CC,\lie_{X}\CC)+2g(\lie_{X}\CC - \lie_{X}\CC,\CC) \\
	&=\lie_{X}g(\CC,\CC) \\
	&=\nabla g(\CC,\CC)+g(\nabla_\CC X,\CC)+g(\CC,\nabla_\CC X) \\
	&=2g(\CC,\nabla_\CC X)
	\end{split}
	\]
	This way, we see that neither of $X$ being Killing or $\lie_X g=0$ implies the other, and additionally we recover the characterization of \cite[Prop. 6.1 (i)]{HJP}.
\end{rem}

\begin{defi}\label{d_divT} Let $L$ be a fixed pseudo-Finsler metric defined on $A$ with (Levi-Civita--)Chern anisotropic connection $\nabla$. For $T\in\mathcal{T}_1^1(M_A)$, its \emph{divergence with respect to $L$} is defined as  
	\[
	\mathrm{div}(T):=\mathbf{C}^1_2(\nabla T)\in\ten_1^0(M_A)=\form_1(M_A),
	\]
	where $\mathbf{C}^1_2$ is the operator that contracts the contravariant index with the covariant one introduced by $\nabla$. In coordinates, 
	\begin{equation}
	\mathrm{div}(T)_j=\nabla_i T^i_j=\delta_i T^i_j+\varGamma_{ik}^iT^k_j-\varGamma_{ij}^kT^i_k
	\label{eq:divergence tensor coordinates}
	\end{equation}
	for the Christoffel symbols of \eqref{eq:chern}.
\end{defi}

\begin{rem}[Divergence vs. raising and lowering indices] \label{rem:divergence tensor}
	\noindent \begin{enumerate} [label=(\roman*)]
		\item First and foremost, by construction, \eqref{eq:div(T(X))} indeed holds for any $X\in\vect(M_A)$. At this point, it is important that the connection with which one defines $\mathrm{trace(\nabla X)}$ is the Chern one. 
		
		\item Thanks to the fact that the Chern connection parallelizes $g$, namely $\nabla_kg_{ij}=0$ and $\nabla_k g^{ij}=0$, the following hold: 
		\begin{equation}
		g^{ik}\nabla_k T_{ij}=g^{ik}g_{il}\nabla_k T^l_{j}=\nabla_k T^k_{j}=\mathrm{div}(T)_j,
		\label{eq:div flat}
		\end{equation}
		\begin{equation}
		\nabla_iT^{ij}=\nabla_iT^i_{l}g^{lj}=g^{jl}\mathrm{div}(T)_l.
		\label{eq:div sharp}
		\end{equation}
		This means that one could define the divergences of $S\in\mathcal{T}_2^0(M_A)$ and $R\in\mathcal{T}_0^2(M_A)$ straightforwardly,\footnote{Here, $\mathbf{C}_{1,3}$ is the operator that (metrically) contracts the first index of $S$ with the one introduced by $\nabla$, and $\mathbf{C}_1^1$ is the operator that (naturally) contracts the first index of $R$ with the one introduced by $\nabla$.} $\mathrm{div}(S)=\mathbf{C}_{1,3}(\nabla S)\in\mathcal{T}_1^0(M_A)=\form_1(M_A)$ and $\mathrm{div}(R)=\mathbf{C}_1^1(\nabla R)\in\mathcal{T}_0^1(M_A)=\vect(M_A)$, and then \eqref{eq:div flat} and \eqref{eq:div sharp} would read respectively 
		\[
		\mathrm{div}(T^\flat)=\mathrm{div}(T),
		\]
		\[
		\mathrm{div}(T^\sharp)=\mathrm{div}(T)^\sharp.
		\]
		
		\item Regardless of this, in general we are not assuming the symmetry of $T^\flat$ or $T^\sharp$, we only did in Prop. \ref{prop:div(T(X))} (Cii). Instead, at the beginning of \S 5 we fixed a convention for the order of the indices in $T_{ij}$ and $T^{ij}$ (for example, $T^\flat(X,Y)=g(X,T(Y))\neq g(T(X),Y)$). In the remainder of \S 4 and with said condition (Cii) only.
	\end{enumerate}
\end{rem}

\subsection{Chern vs. Berwald}\label{d_5.2}

One needs to keep in mind a discussion present in \cite{JSV1}. The metric connection $\HH A$ is the underlying nonlinear connection of an infinite family of $A$-anisotropic connections $\nabla$. One of them is the (Levi-Civita)--Chern connection of $L$, which is the horizontal part of Chern-Rund's and Cartan's classical connections and has Christoffel symbols \eqref{eq:chern}. All the others are this one plus an anisotropic tensor $Q\in\ten_2^1(M_A)$ with $Q(-,\CC)=0$ when viewed as an $\fun(A)$-bilinear map $\vect(M_A)\times\vect(M_A)\rightarrow\vect(M_A)$. In particular, for $Q=-\la^\sharp$, one gets the Berwald anisotropic connection of $L$, which is the horizontal part of Berwald's and Hasiguchi's classical connections and has Christoffel symbols \eqref{eq:berwald}. We did not a priory select any of these $\nabla$'s.

In some of the previous literature \cite{DL,MT,NDT,NNKN}, the Finslerian divergence of vector fields was chosen to be defined directly with the Chern connection. In \cite{Rund,Min17}, the quantity $\mathrm{trace}(\nabla Z)$, with $\nabla$ the Chern anisotropic connection, was referred to as the divergence of $Z$, though only after it had appeared in the divergence theorem. We have proven that the most natural definition leads to this characterization, hence clarifying why using Chern's covariant derivative is not arbitrary. Moreover, we have seen that said derivative fulfills the natural requisite \eqref{eq:div(T(X))} and is compatible with the lowering and raising of indices; these are key properties when it comes to the stress-energy tensor $T$. Still, it is important to compare this with what happens when one uses the other most natural covariant derivative: Berwald's.

\begin{rem}[Divergence in terms of the Berwald connection] Let $\nabla$ be the Chern anisotropic connection of $L$, with Christoffel symbols \eqref{eq:chern}, and $\widehat{\nabla}$ be the Berwald one, with symbols \eqref{eq:berwald}.
	\noindent \begin{enumerate} [label=(\roman*)] 
		\item \eqref{eq:divergence vector coordinates} and \eqref{eq:divergence tensor coordinates} read respectively
		\[
		\di(Z)=\widehat{\nabla}_i Z^i+\la_k Z^k=\mathrm{trace}(\widehat{\nabla} Z)+\la^{\mathfrak{m}}(Z),
		\]	 
		\[
		\begin{split}
		\di(T)_j &= \widehat{\nabla}_i T^i_j + \la_k T^k_j - \la_{ij}^k T_k^i \\
		&= \mathbf{C}^1_2(\widehat{\nabla} T)_j + \la^{\mathfrak{m}}(T)_j - \mathbf{C}^1_1(\la^\sharp(T(-),-))_j,
		\end{split}
		\]
		where $\la^{\mathfrak{m}}$ is the mean Landsberg tensor (see \eqref{eq:mean landsberg}) and the contraction operators have the obvious meanings. Moreover, for $X\in\vect(M_A)$
		\[
		\qquad\begin{split}
		\mathrm{trace}(T(\nabla X)) = T^i_j \nabla_iX^j &= T^i_j \widehat{\nabla}_iX^j + T^i_j \la_{ik}^j X^k \\
		&= \mathrm{trace}(T(\widehat{\nabla} X)) + \mathrm{trace}(\la^\sharp(T(-),X)),
		\end{split}
		\]
		which makes \eqref{eq:div(T(X))} consistent with the previous formulas.
		
		\item One sees that the vanishing of $\la^{\mathfrak{m}}$ (or of the mean Cartan $C^\mathfrak{m}$, see \cite[(6.37)]{Shen2}) implies that the divergence of elements of $\vect(M_A)$ coincides with the trace of their Berwald covariant derivative. However $\la^{\mathfrak{m}} = 0$ (or even $C^\mathfrak{m} = 0$) is not enough if one wants to obtain the same characterization for elements of $\ten_1^1(M_A)$. 
	\end{enumerate}
\end{rem}

\begin{rem}[Sufficient conditions for $\lie_X g=0$ and being Finslerian Killing]
	In Rem. \ref{eq:killing} one could see that $X\in\vect(M)$ together with $\nabla_\CC X=0$ is sufficient for $X$ to be Killing. This condition does not privilege the Chern connection $\nabla$ against the Berwald $\widehat{\nabla}$:
	\[
	\nabla_\CC X=\widehat{\nabla}_\CC X + \la^\sharp(\CC,X) = \widehat{\nabla}_\CC X
	\]
	(see \cite[(38)]{Jav1}, where $\Lie^\flat$ is what here we would denote $\la^\sharp$). However, when it comes to the stress-energy tensor, we have seen that the relevant condition is not this, but rather $\lie_X g=0$. Prop. \ref{prop:div(T(X))} (A) implies that $\nabla^v X=0$ is sufficient for $\left(\lie_X g \right)_v = 0$, and this does privilege $\nabla$ against $\widehat{\nabla}$.
\end{rem}

\subsection{Finslerian conservation laws and main examples}\label{d_5.3}

Compare the results here with the classical case \eqref{eq:riemannian conservation law} and also with \cite{Min17}.

\begin{cor}\label{cor:conservation laws0}
 Let $L$ be a fixed pseudo-Finsler metric defined on $A$. If
	\begin{enumerate} [label=(\roman*)]
		\item $X\in\mathfrak{X}(M_A)$ is an anisotropic vector field, 
		\item $V\in\mathfrak{X}^A(U)$ is an $A$-admissible field with $U\subseteq M$ open, 
		
\item $T\in\ten_1^1(M_A)$	is an anisotropic 2-tensor,	
		and
		\item $\overline{D}\subseteq U$ is a domain with $\partial D$ smooth up to subset of $0$ $\left(n-1\right)$-dimensional measure on $M$ and $\mathrm{Supp}(X_V)\cap\overline{D}$ compact,
	\end{enumerate} 
	then
	\begin{equation}\label{completeintegral}
	\begin{split}
	&\quad\int_D \mathrm{div}(T)(X) d\mathrm{Vol}_V+\int_D\mathrm{trace}(T(\nabla X))_V d\mathrm{Vol}_V\\
	&+\int_D\left\{C^\mathfrak{m}(\mathrm{D}_{T(X)}V)+\mathrm{trace}(\dot\partial _{\mathrm{D}V}T(X))\right\}
	d\mathrm{Vol}_V  =\int_{\partial D}\imath_{T(X)_{V}}(d\mathrm{Vol}_V),
	\end{split}
	\end{equation}
	where $C^\mathfrak{m}$ is the mean Cartan tensor and $\mathrm{D}V$ is computed with the metric nonlinear connection \eqref{eq:metric connection}.
\end{cor}
\begin{proof}
Just  take $Z=T(X)$ in Th. \ref{th:divergence theorem} and use part (B) or Prop. \ref{prop:div(T(X))}	.
\end{proof}
\begin{rem}\label{rem:divT=0}
Observe that  \eqref{completeintegral} allows for an interpretation of the divergence of $T$ in terms of the flow in the boundary. Consider a sequence of domains $D_m$ such that their volumes go to zero when $m\rightarrow +\infty$ and consider an observer $V$ such that is infinitesimally parallel at $p\in M$, namely, $\mathrm{D}V=0$ in $p\in M$ and $X$ such that $\nabla^vX=0$. Then \eqref{completeintegral} and the mean value theorem imply that
\[ \di(T)_v(X) =\lim_{m\rightarrow +\infty} \frac{1}{\mathrm{Vol}_V(D_m)} \int_{\partial D_m}\imath_{T(X)_{V}}(d\mathrm{Vol}_V).	\]
In particular, $\di(T)_v=0$ can be interpreted as that the observer $v$ measures conservation of energy in its restspace.
\end{rem}

\begin{cor} \label{cor:conservation laws} In the ambient of the previous corollary, 
assume:
\begin{enumerate} [label=(\roman*)]
	\item $\di(T)_V=0$.
	
	\item Any of the conditions (Ci) or (Cii) of Prop. \ref{prop:div(T(X))} holds for $T^\flat_V$.
	
	\item $C^\mathfrak{m}(\mathrm{D}_{T(X)}V)+\mathrm{trace}\left\{\dot\partial _{\mathrm{D}V}(T(X))\right\}=0$.
\end{enumerate}
Then
\begin{equation}
	\int_{\partial D}\imath_{T_{V}(X_{V})}(d\mathrm{Vol}_V)=0.
	\label{eq:conservation laws}
\end{equation}
\end{cor}
 
\begin{proof}

{  It follows from Cor \ref{cor:conservation laws0}, taking into account that the hypotheses $(i)$, $(ii)$ and $(iii)$ imply that the three first integrals in \eqref{completeintegral} vanish.  }
\end{proof}

\begin{rem} [Sufficient conditions for the hypotheses (i), (ii) and (iii)]
	\noindent	\begin{enumerate} [label=(\roman*)]
		\item Obviously, $\di(T)=0$ suffices, but we do not need to assume that the divergence vanishes for all observers.
		\item $X=\CC$ suffices. In fact, $\nabla\CC=0$ \cite[Prop. 2.9]{JSV2}, so (Ci) of Prop. \ref{prop:div(T(X))} holds for $T^\flat_V$. Thus, assuming the other two hypotheses, we get 
		\[
		\int_{\partial D}\imath_{T_{V}(V)}(d\mathrm{Vol}_V)=0.
		\] 
		
		\item  Although the hypothesis may seem artificial as it stands, there are a number of natural situations in which it is guaranteed. First, in classical Relativity ($g$, $T$ and $X$ isotropic), because $C^\mathfrak{m}=0$ and $\dot{\partial}(T(X))=0$; the result is then independent of $V$. Second, when the observer field is parallel ($\mathrm{D}V=0$), trivially. Third, when $\mathrm{D}V=\theta\otimes V$ for some $1$-form $V$ and $T(X)$ is $0$-homogeneous, because of Euler's theorem. And fourth, in the situation described in \cite[\S 5.1]{Min17} ($Z$ is our $T(X)$, $s$ is our $V$ and $I$ is our $C^\mathfrak{m}$).
	\end{enumerate}
\end{rem}

\begin{rem} [Representations of \eqref{eq:conservation laws}] \label{rem:conservation laws normals}
	One needs to keep in mind Rem. \ref{rem:normals}. For a smooth part $\Gamma$ of $\partial D$, one can use the (salient) Riemannian unit normal to represent 
	\begin{equation}
	\begin{split}
	\int_{\Gamma}\imath_{T_{V}(X_{V})}(d\mathrm{Vol}_V)&=\int_{\Gamma}g_V(\widehat{N}_V,\widehat{N}_V)g_V(\widehat{N}_V,T_V(X_V))d\sigma_V \\
	&=\int_{\Gamma}g_V(\widehat{N}_V,\widehat{N}_V)T^\flat_V(\widehat{N}_V,X_V)d\sigma_V
	\end{split}
	\label{eq:conservation laws option 1}
	\end{equation}
	when $\Gamma$ is non-$g_V$-lightlike, and the Finslerian unit normal to represent
	\[
	\int_{\Gamma}\imath_{T_{V}(X_{V})}(d\mathrm{Vol}_V)=\int_{\Gamma}\epsilon_\xi L(\xi)g_\xi(\xi,T_V(X_V))d\Sigma^\xi_V
	\]
	when $L$ is Lorentz-Finsler and $\Gamma$ is $L$-spacelike.
	This makes it possible to have the very same conservation law \eqref{eq:conservation laws} written in distinct ways, and in the examples below we will see that different expressions are preferable in different situations.
\end{rem} 

In the remainder of the section, we analyze the Finslerian conservation laws in two settings in which $L$ is Lorentz-Finsler. In particular, $g$ has signature $(+,-,...,-)$, $A$ determines a time orientation, $L>0$ on $A$, and $(A,L)$ is maximal with these properties. We also have regularity conditions at $\partial A$, and in fact one sees that Th. \ref{th:divergence theorem} and Cor. \ref{cor:conservation laws} still hold when allowing that $Z,X\in\vect(M_{\overline{A}})$, $T\in\ten_1^1(M_{\overline{A}})$ and $V\in\vect^{\overline{A}}(U)$. Despite this, in both settings it will be necessary to take $V$ as $L$-timelike, so the regularity at $\partial A$ will not be used.

\subsubsection{Example: Lorentz norms on an affine space} \label{subsec:affine space}
In this example, we shall particularize Cor. \ref{cor:conservation laws} to the easiest Finslerian setting in which we can assure that its hypothesis (iii) holds. Namely, the structure of an affine space automatically provides an infinite number of parallel observer fields, $V\in\vect^A(M)$ with $\mathrm{D}V=0$. 

To be preicse, suppose that $M=E$ is an affine space equipped with a {\em Lorentz norm} on an open conic subset $A_\ast\subseteq \vec{E}\setminus 0$ (a positive pseudo-Minkowski norm with Lorentzian signature in \cite[Def. 2.11]{JS20}). Under the usual identifications, such a norm can be seen as a Lorentz-Finsler $L$ on $A\subseteq \TT E\setminus\mathbf{0}\equiv E\times\left(\vec{E}\setminus 0\right)$ that is independent of the first factor. Consequently, its fundamental tensor is nothing more than a Lorentzian scalar product $g_v$ for each $v\in A_\ast$. The metric nonlinear connection of $L$ coincides with the canonical connection of $E$, hence so do the Chern and Berwald anisotropic connections.\footnote{For instance, it is clear that in affine coordinates the components of the metric spray vanish, so the geodesics are the straight lines of $E$.} This is what implies that the parallel $V\in\vect^A(E)$ correspond exactly to the elements $v\in A_\ast$.

Let us introduce some notation. Given $(p_0,v)\in A$ with $L(v)=1$, we can consider the Lorentzian scalar product $g_v$ and the orthogonal hyperplane $\mathscr{R}:=p_0+\vec{\mathscr{R}}:=p_0+\left\{w\in\vec{E}\colon g_v(v,w)=0\right\}$. We get an isometry $(t,p)\in\R\times\mathscr{R}\mapsto p+tv\in E$, where $\mathscr{R}$ is equipped with $-\left.g_v\right|_\mathscr{R}$ (a Euclidean scalar product), $\R\times\mathscr{R}$ with $\dd t^2+\left.g_v\right|_\mathscr{R}$ (a Lorentzian one) and $E$ with $g_v$. Let $\overline{\Omega}$ be a compact domain of $\mathscr{R}$ with $\partial\Omega\subseteq\mathscr{R}$ smooth up to a null $\left(n-2\right)$-dimensional measure set, and let $\widehat{n}_v$ be its salient unit $\left(-\left.g_v\right|_\mathscr{R}\right)$-normal. Then for $t_0<t_1$, the compact domain $\overline{D}\equiv\left[t_0,t_1\right]\times\overline{\Omega}\subseteq E$ has the required smoothness to apply Cor. \ref{cor:conservation laws}, its boundary is $\partial D=\left\{t_1\right\}\times\overline{\Omega}\cup\left[t_0,t_1\right]\times\partial\Omega\cup\left\{t_0\right\}\times\overline{\Omega}$, and its salient $g_v$-normal is given by
\[
\left.\widehat{N}_v\right|_{\left\{t_1\right\}\times\Omega}=v,\qquad\left.\widehat{N}_v\right|_{\left]t_0,t_1\right[\times\partial\Omega}=\widehat{n}_v,\qquad\left.\widehat{N}_v\right|_{\left\{t_0\right\}\times\Omega}=-v;
\]
\[
g_v(-v,-v)=g_v(v,v)=L(v)=1,
\]
\[
g_v(\widehat{n}_v,\widehat{n}_v)=-\left(-\left.g_v\right|_\mathscr{R}\right)(\widehat{n}_v,\widehat{n}_v)=-1.
\]

\begin{rem}
For a $V\in\vect^A(E)$ identifiable with $v\in A_\ast$, we know that the hypothesis (iii) of Cor. \ref{cor:conservation laws} holds automatically. If (i) and (ii) hold too, then we get \eqref{eq:conservation laws}, for which we can use the representation \eqref{eq:conservation laws option 1}. However, given the nature of the metric ``nonlinear" and Chern ``anisotropic" connections, it is easy to convince oneself that evaluating the result of anisotropic computations on this $V$ is the same as first evaluating on $V$ and then computing with isotropic tensors. For instance $\di(T)_V=\di(T_V)$ and $\left(\lie_Xg\right)_V=\Lie_{X_V}(g_V)$. As a consequence, mathematically we get exactly the same conservation laws as if we just were in the Lorentzian affine space $(E,g_v)$. Physically, though, different observers will measure different momenta.
\end{rem}

\begin{cor} Let $V\in\vect^A(E)$ parallely identifiable with an $v\in A_\ast$. If $T\in\ten_1^1(E_A)$ is such that $\di(T_V)=0$ and $X\in\vect(E_A)$ is such that $T^\flat_V(-,\nabla^V_{-} X)$ is antisymmetric, or $T^\flat_V$ is symmetric and $\Lie_{X_V}(g_V)=0$, then
\begin{equation}
	\begin{split}
	0&=\int_{\left\{t_1\right\}\times\Omega}T^\flat_V(V,X_V)d\sigma_V-\int_{\left\{t_0\right\}\times\Omega}T^\flat_V(V,X_V)d\sigma_V \\
	&\quad-\int_{\left]t_0,t_1\right[\times\partial\Omega}T^\flat_V(\widehat{n}_V,X_V)d\sigma_V,
	\end{split}
	\label{eq:conservation affine space}
\end{equation}
where $d\sigma_V$ is identifiable with the volume form of $-\left.g_v\right|_{\Omega}$ on $\left\{t_\mu\right\}\times\Omega$ and coincides with the volume form of $\left.g_v\right|_{\left]t_0,t_1\right[\times\partial\Omega}$ on $\left]t_0,t_1\right[\times\partial\Omega$.
\end{cor}

Physically, even though Lorentz norms generalize Very Special Relativity \cite{Bogoslovsky}, the {  classical} interpretations of Special Relativity are still valid; we list them for completeness: $v$ is an instantaneous observer at an event $p_0$,  $\vec{\mathscr{R}}$ is its restspace and $\mathscr{R}$ is the {\em simultaneity hyperplane of $v$}, namely the ``universe at an instant, say $t=0$, as seen by $v$". The affine space structure {  allows for a canonical propagation of $v$ to all of the} spacetime. Hence, if $\overline{\Omega}$ is a space region at $t=0$, then $\overline{D}$ is the ``evolution of $\overline{\Omega}$ along the time interval $\left[t_0,t_1\right]$ as witnessed by $v$". \eqref{eq:conservation affine space} expresses that {\em the variation after some time of the total amount of $X_v$-momentum in $\Omega$ is exactly equal to the amount of it that flowed across $\partial\Omega$}. 

\subsubsection{Example: Cauchy hypersurfaces in a Finsler spacetime}

Here we present a construction which manifestly generalizes that of the previous example, again with straightforward physical interpretations, and we find an estimate that allows us to interpret \eqref{eq:conservation affine space} when $\partial\Omega$ is ``at infinity". We will take $V\in\vect^A(U)$ with $U\subseteq M$ open, and we recall that we will assume the hypotheses of Cor. \ref{cor:conservation laws}.

Suppose that the Finsler spacetime $(M,L)$ is {\em globally hyperbolic}. By this, we mean that there is some (smooth, for simplicity) {\em $L$-Cauchy hypersurface} $\mathscr{S}\subseteq M$: every inextensible $L$-timelike curve $\gamma\colon I\rightarrow M$ (thus $\dot{\gamma}(t)\in A$) meets $\mathscr{S}$ exactly once. Let us assume that there are two $L$-spacelike Cauchy hypersurfaces $\mathscr{S}_0,\mathscr{S}_1\subseteq U$ which do not intersect.\footnote{The case when they interesect can be also conisdered by taking into account that, then, the open set $M\setminus J^+(S_1 \cup S_2)$ is still globally hyperbolic and a Cauchy hypersurface $S_3$ of this open subset will be also Cauchy for $M$ (and it will not intersect any of the previous ones).} Then the results of \cite{BerSan} can be automatically transplanted: there exists a foliation by spacelike Cauchy hypersurfaces $M\equiv\R\times\mathscr{S}$ such that $\mathscr{S}_0\equiv\left\{t_0\right\}\times\mathscr{S}$ and $\mathscr{S}_1\equiv\left\{t_1\right\}\times\mathscr{S}$. Taking the Finslerian unit normal $\xi$ to each level $\left\{t\right\}\times\mathscr{S}$ produces an $L$-timelike field $\xi\in\vect^A(M)$. We can take this $\xi$ to be our $V$, but we will not do so for the most part of this example.

Suppose also that $\left\{\overline{\Omega_{0,m}}\right\}$ is an exhaustion by compact domains of $\mathscr{S}_0$, namely $\overline{\Omega_{0,m}}\subseteq\Omega_{0,m+1}$ and $\underset{m\in\N}{\bigcup}\Omega_{0,m}=\mathscr{S}_0$, such that $\partial\Omega_{0,m}\subseteq\mathscr{S}_0$ is smooth a. e. For $p\in\mathscr{S}_0$, let $\gamma_p$ be the integral curve of $V$ starting at $p$, which necessarily meets $\mathscr{S}_1$ at a unique instant $t_p\in\R$. Put 
\[
\Omega_{1,m}:=\underset{p\in\Omega_{0,m}}{\bigcup}\gamma_p(\left\{t_p\right\})\subseteq\mathscr{S}_1,\qquad\Gamma_p:=\gamma_p\left[\min\left\{0,t_p\right\},\max\left\{0,t_p\right\}\right],  
\]  
\[
D_m:=\underset{p\in\Omega_{0,m}}{\bigcup}\Gamma_p\subseteq U,\qquad\Gamma_m:=\underset{p\in\partial\Omega_{0,m}}{\bigcup}\Gamma_p.
\]

\begin{rem}
	By construction,
\end{rem} 
\begin{enumerate} [label=(\roman*)]
	\item $\left\{\overline{\Omega_{1,m}}\right\}$ is again an exhaustion by compact domains of $\mathscr{S}_1$ such that $\partial\Omega_{1,m}=\underset{p\in\partial\Omega_{0,m}}{\bigcup}\gamma_p(\left\{t_p\right\})\subseteq\mathscr{S}_1$ is smooth a. e.
	
	\item $\overline{D_m}$ is a compact domain of $U$ with $\partial D_m=\overline{\Omega_{1,m}}\cup\Gamma_m\cup\overline{\Omega_{0,m}}\subseteq U$ smooth a. e. We do not really need to consider the union of all the $D_m$'s.
\end{enumerate} 	

Next, for $Z\in\vect(M_A)$, we shall give the quantitative decay condition on (some components of) $Z_V$ so that the integral
\[
\int_{\Gamma_m}\imath_{Z_V}(\vol_V)
\] 
vanishes in the limit. The key fact for it will be that $V$ is everywhere tangent to $\Gamma_m$ (this is composed of $\gamma_p$'s). In particular, as $V$ is $g_V$-timelike, so must be $\Gamma_m$.

\begin{rem} \label{rem:auxiliar riemannian}
	The presence of $V$ allows us to define an auxiliar Riemannian metric $h_V$ on $U$ with norm $\left\Vert-\right\Vert_V$, which gives a very natural way of quantifying. Namely, if $\left\{e_0=V_p/F(V_p),e_1,...,e_n\right\}$ is an orthonormal basis for $g_{V_p}$, then we prescribe it to be also $h_{V_p}$-orthonormal; equivalently,
	\[
	h_{V_p}(u,w)=2g_{V_p}(u,\frac{V_p}{F(V_p)})g_{V_p}(w,\frac{V_p}{F(V_p)})-g_{V_p}(u,w).
	\]
	Then, by construction: 
	\begin{enumerate} [label=(\roman*)]
		\item The volume form of $h_V$ coincides with that of $g_V$, namely $\vol_V$.
		\item The salient unit $h_V$-normal to $\Gamma_m$ coincides with the corresponding $g_V$-normal. We denote it by $\widehat{N}_V$, as in \ref{rem:conservation laws normals}. 
		\item The hypersurface volume form of $\Gamma_m$ with respect to $h_V$ coincides with {  the one computed with} $g_V$, namely $d\sigma_V=i_m^\ast(\imath_{\widehat{N}_V}(\vol_V))$ with $i_m\colon\Gamma_m\hookrightarrow U$ the inclusion. Hence we speak just of the {\em hypersurface volume of $\Gamma_m$}, namely $\sigma_V(\Gamma_m)$. As $\widehat{N}_V$ is {  $g_V$-orthogonal to $V$}, and hence $g_V$-spacelike, we can use the representation
		\begin{equation}
		\begin{split}
		\int_{\Gamma_m}\imath_{Z_V}(\vol_V)&=\int_{\Gamma_m}g_V(\widehat{N}_V,\widehat{N}_V)g_V(\widehat{N}_V,Z_V)d\sigma_V \\
		&=-\int_{\Gamma_m}g_V(\widehat{N}_V,Z_V)d\sigma_V.
		\end{split}
		\label{eq:integral timelike boundary}
		\end{equation}
	\end{enumerate}
\end{rem}

Thanks to \eqref{eq:integral timelike boundary} and the fact that $g_V(\widehat{N}_V,V)=0$, we intuitively see that if $Z_V$ is proportional to $V$ at infinity and the hypersurface volume does not grow too much, then the integral will be negligible. To be precise, we require that
\begin{equation}
K_m \sigma_V(\Gamma_m)\longrightarrow 0 \quad \left(m\longrightarrow \infty\right),
\label{eq:decay condition}
\end{equation}
where
\[
\begin{split}
K_m:&=\underset{\Gamma_m}{\max}\left\Vert  Z_V-g_V(Z_V,\frac{V}{F(V)})\frac{V}{F(V)}\right\Vert_V \\
&=\underset{\Gamma_m}{\max}\left\{\sqrt{g_V(Z_V,\frac{V}{F(V)})^2-g_V(Z_V,Z_V)}\right\}.
\end{split}
\]

\begin{cor} \label{prop:conservation cauchy} In the above set-up, let $T\in\ten_1^1(M_A)$, $X\in\vect(M_A)$ and $V\in\vect^A(U)$ be such that the hypotheses of Cor. \ref{cor:conservation laws} hold on all the $D_m$'s, and put $Z:=T(X)$. If the decay condition \eqref{eq:decay condition} holds too, then
\begin{equation}
\int_{\Omega_{1,m}}\imath_{Z_V}(\vol_V)+\int_{\Omega_{0,m}}\imath_{Z_V}(\vol_V)\longrightarrow 0 \quad (m\longrightarrow\infty),
\label{eq:conservation law cauchy 2}
\end{equation}	
where $\Omega_{1,m}$ is constructed from $\Omega_{0,m}$ by intersecting the integral curves of $V$ with $\mathscr{S}_1$.
\end{cor}

\begin{proof}
	Cor. \ref{cor:conservation laws} can be applied on $\overline{D_m}$, as $\mathrm{Supp}(Z_V)\cap\overline{D_m}$ is always compact. This and the representation \eqref{eq:integral timelike boundary} give 
	\begin{equation}
		0=\int_{\Omega_{1,m}}\imath_{Z_V}(\vol_V)+\int_{\Omega_{0,m}}\imath_{Z_V}(\vol_V)-\int_{\Gamma_m}g_V(\widehat{N}_V,Z_V)d\sigma_V.
		\label{eq:conservation law cauchy 1}
	\end{equation}
	Using the definition of $h_V$ (Rem. \ref{rem:auxiliar riemannian}) and the Cauchy-Schwarz inequality,
	\[
	\begin{split}
	0&\leq\left|\int_{\Gamma_m}-g_V(\widehat{N}_V,Z_V)d\sigma_V\right| \\
	&\leq\int_{\Gamma_m}\left|g_V(\widehat{N}_V,Z_V)\right|d\sigma_V \\
	&=\int_{\Gamma_m}\left|g_V(\widehat{N}_V,Z_V-g_V(Z_V,\frac{V}{F(V)})\frac{V}{F(V)})\right|d\sigma_V \\
	&=\int_{\Gamma_m}\left|-h_V(\widehat{N}_V,Z_V-g_V(Z_V,\frac{V}{F(V)})\frac{V}{F(V)})\right|d\sigma_V \\
	&\leq\int_{\Gamma_m} \left\Vert\widehat{N}_V\right\Vert_V \left\Vert Z_V-g_V(Z_V,\frac{V}{F(V)})\frac{V}{F(V)}\right\Vert_V d\sigma_V \\
	&=\int_{\Gamma_m} \left\Vert Z_V-g_V(Z_V,\frac{V}{F(V)})\frac{V}{F(V)}\right\Vert_V d\sigma_V \\
	&\leq\int_{\Gamma_m} K_m d\sigma_V \\
	&=K_m \sigma_V(\Gamma_m),
	\end{split} 
	\]
	so if $K_m \sigma_V(\Gamma_m)$ tends to $0$, then so does the integral along $\Gamma_m$ in \eqref{eq:conservation law cauchy 1}.
\end{proof}

\begin{rem} \label{rem:assymptotic energy conservation} In Cor. \ref{prop:conservation cauchy}, if one of the integrals of $\imath_{Z_V}(\vol_V)$ along $\mathscr{S}_0$ or $\mathscr{S}_1$ exists in the Lebesgue sense, then so does the other and \eqref{eq:conservation law cauchy 2} reads	
\[
\int_{\mathscr{S}_1}\imath_{Z_V}(\vol_V)+\int_{\mathscr{S}_0}\imath_{Z_V}(\vol_V)=0. 
\]
Note that they could be $\pm\infty$, as we have not assumed, for instance, that $Z_V$ is compactly supported in the union of all the $D_m$'s. Rather, we have assumed the decay condition \eqref{eq:decay condition} alone.
\end{rem}

\begin{rem}[Sufficient conditions for \eqref{eq:decay condition}] \label{rem:decay condition}
	As for ensuring the decay condition, there are two possible scenarios.
	\begin{enumerate} [label=(\roman*)]
		\item The hypersurface volume $\sigma_V(\Gamma_m)$ stays bounded. Then, it is enough for \eqref{eq:decay condition} that $K_m\rightarrow 0$, and one could instead postulate the stronger condition that the maximum outside $D_m$ tends to $0$, which is independent of the concrete compact exhaustion.
		
		\item $\sigma_V(\Gamma_m)$ grows without bound. In this case, one can just postulate that the decay of $K_m$ compensates the growth of $\sigma_V(\Gamma_m)$, but this does depend on the compact exhaustion
	\end{enumerate}
	Notice that this is a purely Finslerian difficulty. Indeed, suppose that $g$, $T$ and $X$ were isotropic and that $Z=T(X)$ was timelike. Then one could just set $V:=Z$ and then carry out all the construction. Cor. \ref{cor:conservation laws} would be independent of the observer field (and its hypothesis (iii) would hold trivially), and $K_m=0$ regardless of $\Gamma_m$. This is how we get the following statement of the classical law.
\end{rem}

\begin{cor} \label{prop:conservation cauchy classical}In the above se-up, suppose that $L$ comes from a Lorentzian metric on $M$. Let $T\in\ten_1^1(M)$ and $X\in\vect(M)$ be such that $\di(T)=0$ and $T^\flat(-,\nabla_{-} X)$ is antisymmetric, or $T^\flat$ is symmetric and $\Lie_{X}g=0$. If $Z:=T(X)$ is timelike, then
	\[
	\int_{\Omega_{1,m}}\imath_{Z_V}(\vol_V)+\int_{\Omega_{0,m}}\imath_{Z_V}(\vol_V)\longrightarrow 0 \quad (m\longrightarrow\infty), 
	\]
where $\Omega_{1,m}$ is constructed from $\Omega_{0,m}$	by intersecting the integral curves of $Z$ with $\mathscr{S}_1$.
\end{cor}

\begin{rem}[Conservation in terms of the Finslerian unit normal]
	\noindent \begin{enumerate} [label=(\roman*)]
		\item One could try to represent also the integrals of \eqref{eq:conservation law cauchy 2} in terms of $d\sigma_V$, as in \S \ref{subsec:affine space}. However, according to Rem. \ref{rem:conservation laws normals}, that would require assuming that $\mathscr{S}_\mu$ is non-$g_V$-lightlike, which is not very reasonable when all we know is that $\mathscr{S}_\mu$ $L$-spacelike and $L$-Cauchy.

		\item 
		On the other hand, in terms of the Finslerian unit normal $\xi$, \eqref{eq:conservation law cauchy 2} reads 
		\begin{equation}	
				\int_{\Omega_{1,m}}g_{\xi}(\xi,T_V(X_V))d\Sigma_V^\xi-\int_{\Omega_{0,m}}g_{\xi}(\xi,T_V(X_V))d\Sigma_V^\xi\longrightarrow 0
			\label{eq:conservation law cauchy 3}
		\end{equation}
		when $m\rightarrow\infty$. The sign in front of the second integral is explained as follows (see Rem. \ref{rem:normals} (ii)). $d\Sigma_V^\xi$ selects an orientation on each $\Omega_{\mu,m}$: the one for which $\vol_V(\xi,-,...,-)$ is positive. However, in \eqref{eq:conservation law cauchy 2} $\Omega_{1,m}$ already had an orientation $\mathfrak{O}_1$ and $\Omega_{0,m}$ had $\mathfrak{O}_0$: the $D_m$-salient ones. Necessarily,\footnote{Suppose, for instance, that $\mathscr{S}_{1}$ lays in the future of $\mathscr{S}_{0}$: the $\gamma_p$'s departing from $\Omega_{0,m}$ reach points $\gamma_p(t_p)\in\Omega_{1,m}$ with $t_p>0$. Take bases $(e_1,...,e_{n-1})$ for $\TT_{p}\Omega_{0,m}$ and $(e^\prime_1,...,e^\prime_{n-1})$ for $\TT_{\gamma_p(t_p)}\Omega_{1,m}$ such that $(V_p,e_1,...,e_{n-1})$ and $(V_{\gamma_p(t_p)},e^\prime_1,...,e^\prime_{n-1})$ are $\vol$-positive. Then $(e_1,...,e_{n-1})$ and $(e^\prime_1,...,e^\prime_{n-1})$ are both $d\Sigma_V^\xi$-positive ($\xi$ and $V$ always lie in the same half-space), the former is $\mathfrak{O}_0$-negative ($V$ is $D_m$-entering at $\mathscr{S}_{0}$) and the latter is $\mathfrak{O}_1$-positive ($V$ is $D_m$-salient at $\mathscr{S}_{1}$).} exactly one of these agrees with the $d\Sigma_V^\xi$-orientation: $\mathfrak{O}_1$ if $\mathscr{S}_{1}$ lays in the future of $\mathscr{S}_{0}$ and $\mathfrak{O}_0$ if it is the opposite. Notice that this, and hence \eqref{eq:conservation law cauchy 3}, would fail if the Cauchy hypersurfaces crossed.
		
		\item In the case $V=\xi$, \eqref{eq:conservation law cauchy 3} becomes 
		\[
		\int_{\Omega_{1,m}}T_\xi^\flat(\xi, X_\xi)d\Sigma_\xi-\int_{\Omega_{0,m}}T_\xi^\flat(\xi, X_\xi)d\Sigma_\xi\longrightarrow 0,
		\]
		a conservation law in which all the terms are purely Finslerian.
	\end{enumerate}
\end{rem}

Summing up, in this example we have proven a Finslerian (observer-dependent) version of the classical law that {\em the total amount of $X$-momentum in the universe is conserved} (Cor. \ref{prop:conservation cauchy}). Our formulation is asymptotic, so it is valid even for infinite total $X_V$-momentum (Rem. \ref{rem:assymptotic energy conservation}). We have recovered the classical law (Cor. \ref{prop:conservation cauchy classical}), which always holds under hypotheses on $T$ and $X$ alone, while in the general Finslerian case nontrivial difficulties appear in the regime of big separation between the Cauchy hypersurfaces (high $\sigma(\Gamma_m)$, Rem. \ref{rem:decay condition}). Finally, we have expressed the law naturally in terms of the Finslerian unit normal (see \eqref{eq:conservation law cauchy 3}).

\section{Conclusions}

About the physical interpretation of $T$, \S \ref{s3}:
\begin{enumerate}
\item {\em Heuristic interpretations from fluids}, \S \ref{s3.1}  and \ref{s3.2} Possible  breakings of Lorentz-invariance lead to non-trivial transformations of coordinates between observers. Such transformations are still linear and permit a well-defined energy-momentum vector at each tangent space $\TT_pM$, \S \ref{s3.1}. 

However, 
the stress-energy-momentum $T$ must not be regarded as a tensor on each $\TT_pM$, but as an anisotropic tensor.  This depends intrinsically on each observer $u\in \Sigma$ and may vary with $u$ in a nonlinear way. 
Indeed, the breaking of Lorentz invariance does not permit to fully replicate the relativistic arguments leading to (isotropic) tensors on $M$, even though classical interpretations of the anisotropic $T$ in terms of fluxes can be maintained, \S \ref{s3.2}. 

\item  {\em Lagrangian viewpoint},
\S \ref{s_lagr}.  In principle, the interpretations of Special Relativity about the canonical energy-momentum tensor associated with the invariance by translations remain for Lorentz norms and, thus, in Very Special Relativity. In the case of Lorentz-Finsler metrics, some issues to be studied further appear: 

\ben \item The canonical stress-energy tensor in Relativity $\delta S_{matter}/ \delta  g^{\mu\nu}$ leads  to different types of (anisotropic) tensors in the Finslerian setting (a scalar function $\delta S_{matter}/ \delta  L$  on $A\subseteq \TT M$ in the Einstein-Hilbert setting,  higher order tensors in Palatini's). Starting at such tensors, different alternatives to recover the heuristic physical interpretations in terms of a 2-tensor appear. 
\item     In the particularly interesting case of a kinetic gas \cite{HPV0, HPV21}, the 1-PDF $\phi$ becomes naturally the matter source for the Euler-Lagrange equation of the Finslerian Einstein-Hilbert functional. However, the variational derivation of $\phi$ is obtained by means of a non-natural Lagrangian. This might be analyzed {  by sharpening} the framework of variational completion for Finslerian Einstein equations \cite{HPV}.  
\een
\end{enumerate}

 About the divergence theorem for anisotropic vector fields $Z$, \S \ref{sec:divergence vectors}: 

\ben
\item \S \ref{subsec:bracket cartan}: For any Lorentz Finsler metric $L$, there is a natural definition of {\em anisotropic Lie bracket derivation along $Z$}, which depends only on the nonlinear connection $\HH A$ and admits an interpretation by using flows.

\item \S \ref{s_4.2}: This bracket allows one to give a natural definition of $\di(Z)$ which depends exclusively on  $\HH A$ 
and the volume {  form of} $L$. This provides a geometric interpretation for the definition of divergence introduced by Rund \cite{Rund}.

\item A general divergence theorem is obtained (Th. \ref{th:divergence theorem}) so that \S \ref{sec:divergence tensor}:
\ben\item   It can be seen as a conservation law for $Z$ {  measured} by each {  observer field} $V$, even if the conserved quantity  depends on $V$. 

\item \label{(i)} The computation of the boundary term is intrinsically expressed in terms of forms. 
However, several metric elements can be used to re-express it, in particular the normal vector  field  for: 

(i) the {  pseudo-Riemannian metric} $g_V$ (Rund), or 
(ii) the pseudo-Finsler metric $L$, when $L$ is defined on the whole $\TT M$ (Minguzzi). 

\een

\een

 About the conservation of the stress-energy $T$ \S \ref{s5}:
\ben \item \S \ref{d_5.1} and \ref{d_5.2}: The computation of $\di(T)$ priviledges the Levi-Civita--Chern anisotropic connection, showing explicit equivalence with Rund's {approach}. 

\item  Cors. ~\ref{cor:conservation laws0} and ~\ref{cor:conservation laws}: A vector field $T(X)_V$ on $M$ is preserved assuming that some natural elements vanish  on $V$ for $T$, $X$ and  $\mathrm{D}V$.

\item \S \ref{d_5.3}: Natural laws of conservation on Cauchy hypersurfaces under general conditions (including rates of decay for unbounded domains) can be obtained by a  combination of the techniques (i) and (ii) in
 \ref{(i)}.  
\een

\section{Appendix. Kinematics: observers and relative velocities } \label{s_A}

Here, we discusss  a series of different possibilities for the notion of relative velocity between two observers, each one with a well-defined geometric construction.
This is done as an academic exercise, because we do not  discuss     experimental issues 
  (compare with \cite{LP, Pfeifer}). However, it is worth emphasizing that all the possibilities studied here are intrinsic to the geometry of a flat model and, thus to any Finsler spacetime. 

Start at an affine space endowed with a Lorentz norm let $u, u'\in \Sigma$ be two distinct observers 
and consider the plane $\Pi:= $ Span$\{u,u'\} \subset V$, which intersects transversally $\C$ and inherits a Lorentz Finsler norm with indicatrix $\Sigma_\Pi:= \Pi \cap \Sigma$. Recall  that both tangent spaces $\TT_u \Pi$ and $\TT_{u'}\Pi$ inherit naturally a Lorentz scalar product by restricting the fundamental tensors $g_u$ and $g_{u'}$, 
resp. Moreover, their (1-dimensional) restspaces $l:= \TT_u\Sigma_\Pi$, $l':= \TT_{u'}\Sigma_\Pi$ also inherit a positive definite metric. In what follows, only the geometry of $\Pi$ will be relevant.

\subsection*{The Lorentz metric $g_\Pi$ up to a constant}\label{s3.1.1}
Notice that $\Pi \cap \C_p$ is composed by two half-lines spanned by two  $\C$-lightlike directions $w_\pm$; we will consider the orientation $\Pi$ provided by the choice $(w_+,w_-)$. One can determine a 
scalar  product $g_\Pi$ in $\Pi$ (which is  unique up to a positive constant), regarding both $w_+$ and $w_-$ as $g_\Pi$-lightlike in the  same causal cone.   It is easy to check that  $\Sigma$ must be a strongly convex curve which converges asymptotically to the vector lines spanned by $w_\pm$.
This implies both  $u\in \Sigma$   will be timelike for $g_\Pi$ and its restpace $l$ will be  $g_\Pi$-spacelike; we can assume also  that the orientation  $l_+$ in $l$ is induced by the chosen $w_+$. 

Notice that $g_u(u,w_\pm)\geq 0$ by the fundamental  inequality,  but  $w_\pm$ might be timelike or spacelike for $g_u$ (although $g_u(u,w_\pm)\rightarrow 0$ as $u\rightarrow w_\pm$).
This possibility might be regarded as a possible measurement of the speed of light with respect to $u$ by the observers in $\Pi$, namely, this velocity is in the orientation $l_+$  when $w_+$ is $g_u$-spacelike and smaller than 1 when it is timelike. However, a priori it is not clear an operational way to carry out such a measurement. Moreover such a measurement might be regarded as something non-intrinsic to the speed of light but to the way  of measuring  it. 

Nevertheless, as pointed out

in  \cite[Section 6]{BerJavSan}, there are several effects which might lead to a measurement of different speeds of light in different directions. So, we will consider that each $\Pi$ has its own speeds of light $c_{\Pi}^\pm$ in each spacelike orientation $l_\pm$.

 Indeed,  given $u$ and an orientation $l^+$, the {\em speed of light $c_{\Pi}^+$} will be defined as the the supremum of the relative velocities between $u$ and all the observers $u'$ such that $u'-u$ yields the orientation $l^+$. Next, we will explain several possible meanings of these  velocities.   To avoid cluttering, next we will write $c_{\Pi}$, assuming that the appropriate choice in $c_{\Pi}^\pm$  is done for each $u'$.

\subsection*{Simple relative velocity} As $g_u$ determines naturally a Lorentz metric on $V$, we can define the {\em simple relative velocity} $v_u^s(u')$ of $u'$ measured by $u$ as the usual $g_u$-relativistic velocity between $u$, $u'$ normalized to $c_\Pi$, i.e. 
$$v_u^s(u')= c_\Pi \tanh (\theta) \quad \hbox{where} \quad \cosh\theta =  -g_u(u,u')  > 1,$$
(the latter  by the reversed fundamental inequality). Clearly, $v_{u'}^s(u)\neq v_u^s(u')$ in general, but this does not seem a drawback in the Finslerian setting.

A support for the  physical plausibility of this velocity is that one could expect that each observer $u$ will work as in Special Relativity just choosing an orthonormal frame of $g_u$.  The possibility $g_u(v,v)\neq 1$ might seem ackward  from a dynamical viewpoint (see below),  but it seems harmless as far as only kinematics is being considered. In principle, the comparison between the measurements  of the two  observers would be geometrically possible by using the unique isometry of $(\TT_u\Pi , g_u)$ to $(\TT_{u'}\Pi , g_{u'})$ which maps $u$ into $u'$ and is consistent with  orientations induced from $\Pi$.  What is more, this isometry can also be extended to a  natural isometry from $(\TT_u V , g_u)$ to $(\TT_{u'} V , g_{u'})$, namely, regard $(\Sigma,g)$ as a Riemannian metric and use the parallel transport from $u$ to $u'$ along the segment of the curve $\Pi\cap \Sigma$ from $u$ to $u'$. However, the following fact might suggest  to explore further possibilities. 
 
 \begin{rem} \label{r1} Assume that $\Sigma$ is modified into the indicatrix $\bar \Sigma$ of another Lorentz-Finsler norm so that (i) $\bar \Sigma= \Sigma$ around $u$ and (ii)  $u'\in \bar \Sigma$ but its $\bar \Sigma$ restspace $\bar l'$ is different  from  $l'$. Then, the simple velocity would remain unaltered, i.e., $\bar v_u^s(u')= v_u^s(u')$.\end{rem}

\subsection*{Velocity as a distance between observers} Notice that $\Sigma$ can be regarded as a Riemannian manifold with the restriction of the fundamental tensor $g$ and, then, $\Sigma \cap \Pi$ can be regarded as a curve whose length can be computed. Then, the {\em observers' distance velocity} is defined as:
$$
v^d(u,u')= c_\Pi \, \tanh\left(\hbox{length}_g\{\hbox{segment of} \, \Sigma \cap \Pi \, \hbox{from $u$ to $u'$}\}\right). 
$$
Notice that this velocity is symmetric and it  generalizes  directly the one in Special Relativity providing a geometric interpretation for the addition of velocities. Recall that  $v^d(u,u')$ has been defined  essentially  as a distance in  $\Sigma \cap \Pi$, where $\Pi$ depends of each pair of observers, thus, one might have $v^d(u,u')+v^d(u',u'') < v^d(u,u'')$ when $n>2$. 
If one prefers to avoid such a possibility, it is enough to consider  $g$-distance in the whole space of observers $\Sigma$ ({\em  observers' space  
distance velocity}), at least in the case that $c_\Pi$ is regarded as independent of $\Pi$.

\begin{rem} \label{r2}
In the case studied in Remark \ref{r1}, one would have $\bar v^d(u,u')\neq v^d(u,u')$ in general. However, the relative position of the restspaces $l$ and  $l'$ does not play any special role. 
\end{rem}
\subsection*{Length-contraction and velocity} Consider a segment $S$ of $l$ with $g_u$-length $\ell$ and the strip of $V$ obtained by translating $S$ in the direction of $u$.  Let $S'$ be the intersection of this strip with $l'$, which will be a new segment of $g_{u'}$-length $\ell '$.  Let $\lambda = \ell '/\ell$  be   the {\em length-contraction parameter}. In the relativistic case, $\lambda<1$ and $\lambda \rightarrow 0$ as $u' \rightarrow \C_\Pi$. The former property does not hold for a general Lorentz 
norm   but the latter does.    So, whenever $\lambda<1$ holds, we can define the {\em length-contractive velocity  $v^c_u(u')$ of $u'$ with respect to} $u$ as:
$$
v^c_u(u')= c_\Pi \sqrt{1-\lambda^2}.
$$ 
Again, this velocity is not symmetric. Because of the strong convexity of $\Sigma$,  a different observer $u'$ will have a different restspace  $l'$, but this does not imply a different length $\ell'$ nor velocity $v_u^c(u')$.  However, this velocity gives a comparison between restspaces which was absent in the previous two velocities.

\subsection*{Symmetric Lorentz velocities in $\Pi$}\label{ss3.1.5}

Let us consider the Lorentzian  scalar   product $g_\Pi$ en $\Pi$,  unique up to a positive constant (which will be irrelevant for our purposes) introduced   above. 
Recall that  $u$ and $u'$ were timelike for $g_\Pi$ and, moreover, both $l$ and $l'$ were spacelike. Now, we can define two velocities between $u$ and $u'$: the {\em simple Lorentz velocity},  

$$v^s(u,u')= c_\Pi \tanh (\theta) \quad \hbox{where} \quad \cosh\theta = -\frac{g_\Pi(u,u')}{\sqrt{g_\Pi(u,u) g_\Pi(u',u')}},$$
and the {\em length-contractive Lorentz velocity},  
$$v^c(u,u')= c_\Pi \tanh (\theta) \quad \hbox{where} \quad \cosh\theta = -\frac{|g_\Pi(n,n')|}{\sqrt{g_\Pi(n,n) g_\Pi(n',n')}},$$
where, in the latter,  $n$, $n'$ are $g_\Pi$-timelike vectors orthogonal to $l$, $l'$, resp.

Clearly, both velocities are symmetric. Their appearance might be  physically sound because the intrinsic  Lorentz metric $g_\Pi$ (up to a constant) can be regarded as an object available (or, at least,  a compromise one) for all the observers, as it would depend directly on physical light rays.

\section*{Acknowledgments}
MAJ was partially supported by the project  PGC2018-097046-B-I00 funded by MCIN/ AEI /10.13039/501100011033/ FEDER ``Una manera de hacer Europa'' and Fundaci\'on S\'eneca project with reference 19901/GERM/15.  This work is a result of the activity developed within the framework of the Programme in
	Support of Excellence Groups of the Regi\'on de Murcia, Spain, by Fundaci\'on S\'eneca, Science and Technology Agency of the Regi\'on de Murcia. MS and FFV were partially supported by 
the project  PID2020-116126GB-I00 funded by MCIN/ AEI /10.13039/501100011033, by the project PY20-01391 (PAIDI 2020)  funded by Junta de Andaluc\'{\i}a---FEDER and by the framework of IMAG-Mar\'{\i}a de Maeztu grant CEX2020-001105-M funded by MCIN/AEI/ 10.13039/50110001103.

\end{document}